\documentclass{article}

\usepackage{amsmath,amssymb,amsthm}
\usepackage{fullpage}
\usepackage{xcolor,xspace}
\usepackage{graphicx}
\usepackage{caption}
\usepackage{booktabs}
\usepackage{microtype}
\usepackage{enumitem}
\usepackage{thmtools}
\usepackage{thm-restate}
\usepackage{algorithm}
\usepackage[
    indLines=true,
    noEnd=true,
    rightComments=true,
    italicComments=true,
]{algpseudocodex}
\algrenewcommand\algorithmicrequire{\textbf{Input:}}
\algrenewcommand\algorithmicensure{\textbf{Output:}}

\usepackage{pgfplots}
\pgfplotsset{compat=1.18}

\declaretheorem[name=Theorem]{theorem}
\declaretheorem[name=Lemma,sibling=theorem]{lemma}

\declaretheorem[name=Definition,sibling=theorem,style=definition]{definition}

\usepackage[
    colorlinks=true,
    linkcolor=blue!62!black,
    citecolor=green!48!black,
    urlcolor=blue!70!black,
    linktoc=page
]{hyperref}
\usepackage[nameinlink,noabbrev]{cleveref}

\crefname{theorem}{Theorem}{Theorems}
\Crefname{theorem}{Theorem}{Theorems}
\crefname{lemma}{Lemma}{Lemmas}
\Crefname{lemma}{Lemma}{Lemmas}
\crefname{corollary}{Corollary}{Corollaries}
\Crefname{corollary}{Corollary}{Corollaries}
\crefname{observation}{Observation}{Observations}
\Crefname{observation}{Observation}{Observations}
\crefname{proposition}{Proposition}{Propositions}
\Crefname{proposition}{Proposition}{Propositions}

\crefname{claim}{Claim}{Claims}
\Crefname{claim}{Claim}{Claims}
\crefname{conjecture}{Conjecture}{Conjectures}
\Crefname{conjecture}{Conjecture}{Conjectures}
\crefname{assumption}{Assumption}{Assumptions}
\Crefname{assumption}{Assumption}{Assumptions}
\crefname{definition}{Definition}{Definitions}
\Crefname{definition}{Definition}{Definitions}
\crefname{remark}{Remark}{Remarks}
\Crefname{remark}{Remark}{Remarks}
\crefname{algorithm}{Algorithm}{Algorithms}
\Crefname{algorithm}{Algorithm}{Algorithms}
\crefname{section}{Section}{Sections}
\Crefname{section}{Section}{Sections}
\crefname{appendix}{Appendix}{Appendices}
\Crefname{appendix}{Appendix}{Appendices}

\tikzset{
  algpxIndentLine/.style={draw=black!100,very thin}
}

\newcommand{\R}{\mathbb{R}}
\newcommand{\eps}{\varepsilon}
\newcommand{\1}{\mathbf{1}}
\newcommand{\cU}{\mathcal{U}}
\newcommand{\cL}{\mathcal{L}}
\newcommand{\ip}[2]{\left\langle #1,#2\right\rangle}
\newcommand{\norm}[1]{\left\lVert #1\right\rVert}
\newcommand{\vol}{\operatorname{vol}}
\newcommand{\tr}{\operatorname{Tr}}

\title{Deterministic Spectral Sparsification in Almost-Linear Time for Dense Graphs}
\author{
    Jason Li\thanks{Carnegie Mellon University. Email: jmli@cs.cmu.edu}
    \and
    Trevor Vaughn\thanks{Carnegie Mellon University. Email: tnvaughn@cmu.edu}
}

\begin{document}
\maketitle

\begin{abstract}
A spectral sparsifier of a weighted graph is a reweighted subgraph whose
Laplacian quadratic form approximates that of the original graph. Let $G$
be a positively weighted $n$-vertex, $m$-edge multigraph, let $0<\varepsilon\le1/2$.
Assuming $m,\varepsilon^{-1}\le n^{O(1)}$ and the ratio of maximum to minimum weight is polynomially bounded, we deterministically construct a
$(1\pm\varepsilon)$-spectral sparsifier with
\[
 O\!\left(n\varepsilon^{-2}\log^{24+o(1)}n\right)
\]
edges in
\[
 m^{1+o(1)}+O\!\left(n^2\varepsilon^{-9/2}\log^{113/2+o(1)}n\right)
\]
time.

The construction has two main ingredients. First, we sparsify an approximately regular expander by partitioning its edges into few matchings and
viewing their normalized Laplacians as an isotropic family of positive
semidefinite matrices. Rather than sample from this family and apply
matrix Chernoff, we select matchings deterministically using a pessimistic
estimator. We evaluate the resulting conditional-expectation scores
in two ways to produce two algorithms: using dense matrix multiplication and sparsely using polynomial
approximations to the inverse square root and matrix exponential. Deterministic expander decomposition,
along with replacing vertices by fixed expander graphs to achieve approximate regularity, extends these algorithms to general
graphs. Second, a recursive blocking scheme applies the dense algorithm
to smaller subgraphs and the sparse algorithm to their union, balancing
their costs. Reusing the resulting algorithm as the dense algorithm gives
$\alpha_{r+1}=3-1/(\alpha_r-1)$, starting from $\alpha_0=\omega$.
After $O(\log n)$ levels, the exponent is $2+O(1/\log n)$, yielding
$m^{1+o(1)}+\widetilde O_{\varepsilon}(n^2)$ time.
\end{abstract}

\section{Introduction}

A spectral sparsifier replaces a graph by a much sparser reweighted
subgraph without substantially changing its Laplacian.  More precisely,
a weighted subgraph \(H\) of \(G\) is a
\((1\pm\eps)\)-spectral sparsifier if
\[
    (1-\eps)L_G\preceq L_H\preceq(1+\eps)L_G.
\]
Equivalently, \(x^\top L_Hx\) approximates \(x^\top L_Gx\) simultaneously
for every vector \(x\).  This contains approximate preservation of every cut as the
special case in which \(x\) is an indicator vector, while also preserving
the linear-algebraic structure used in electrical flows, random walks,
Laplacian solvers, and numerous other graph algorithms.

Spielman and Teng~\cite{ST11} introduced spectral sparsification and gave
the first nearly-linear-time constructions.  Spielman and
Srivastava~\cite{SS11} developed effective-resistance sampling and showed
that \(O(n\eps^{-2}\log n)\) edges suffice.  Batson, Spielman, and
Srivastava~\cite{BSS12} removed the logarithmic loss, proving that every
graph has a spectral sparsifier with \(O(n\eps^{-2})\) edges and giving a
deterministic polynomial-time construction.

There has since been substantial progress on constructing optimal-size
sparsifiers quickly in the randomized setting.  Allen-Zhu, Liao, and
Orecchia~\cite{AZLO15} introduced a regret-minimization formulation and an
almost-quadratic implementation.  Lee and Sun~\cite{LS18}
gave a randomized algorithm producing \(O(n\eps^{-2})\) edges in
\(\widetilde O(m\eps^{-O(1)})\) time.  More recently, Jambulapati, Reis,
and Tian~\cite{JRT24} used efficient discrepancy minimization to obtain
the same optimal sparsity in
\(\widetilde O(m\eps^{-7/2})\) randomized time.  Thus randomized
constructions can simultaneously achieve essentially optimal output size
and near-input-sparsity running time.

The deterministic picture is different.  The barrier argument of
Batson, Spielman, and Srivastava gives optimal sparsity but is slow,
and recent deterministic discrepancy-walk methods are also large polynomial-time constructions~\cite{LWZ25}.

\paragraph{Our result.}
Let \(G\) be a positively weighted \(n\)-vertex, \(m\)-edge multigraph, and
let
\(b=1+\lceil\log_2(w_{\max}/w_{\min})\rceil\).
For \(m,\eps^{-1}\le n^{O(1)}\), we deterministically construct a
\((1\pm\eps)\)-spectral sparsifier with
\[
    O\!\left(
        n(b+\log^2 n)\eps^{-2}\log^{22+o(1)}n
    \right)
\]
edges in
\[
    m^{1+o(1)}
    +
    O\!\left(
        n^2(b+\log^2 n)\eps^{-9/2}
        \log^{109/2+o(1)}n
    \right)
\]
time.  For polynomially bounded weight ratio, this is
\(O(n\eps^{-2}\log^{24+o(1)}n)\) edges in
\[
m^{1+o(1)}+
O\!\left(n^2\eps^{-9/2}\log^{113/2+o(1)}n\right)
\]
 time.

\paragraph{Sparsifying one expander.}
The central part of the algorithm is easiest to understand when \(G\) is
an approximately regular (in terms of combinatorial degree) expander.  Partition its non-loop edges into few matchings
\(M_1,\ldots,M_q\), and let
\(P=(\cL_G|_{\cU_G})^{-1/2}\), where $\cU_G$ is orthogonal to the normalized Laplacian's kernel.  Define
\[
    A_i
    :=
    qP D_G^{-1/2}L_{M_i}D_G^{-1/2}P.
\]
Because the matchings partition the edges,
\(\frac1q\sum_iA_i=I_{\cU_G}\).  Thus the graph has been converted into
an isotropic average of positive semidefinite matrices.

The expander condition controls the size of every matrix in this average.
Cheeger's inequality gives
\(\norm{P}=O(\phi^{-1})\), while approximate regularity and the matching
condition control
\(D_G^{-1/2}L_{M_i}D_G^{-1/2}\).  Together these imply
\(A_i\preceq O(\phi^{-2})I_{\cU_G}\).  Sampling
\(O(\phi^{-2}\eps^{-2}\log n)\) matrices from the average would therefore
give a spectral approximation by matrix Chernoff.  Translating the sampled
matrices back to matchings yields a sparse approximation of the expander.

\paragraph{Derandomizing the matching samples.}
We replace random sampling by conditional expectation.  The matrix
Chernoff proof has a trace-exponential pessimistic estimator that bounds
the probability that the selected average has an eigenvalue that
is too large or too small.  At each iteration,
the average conditional value of this estimator over all possible next
matchings does not increase.  We therefore choose a matching whose
conditional value is no larger than the average.  Repeating this choice
produces deterministically the same spectral guarantee supplied by the
randomized Chernoff argument.

The main algorithmic question is how to evaluate the conditional scores.
All of the matchings' scores are
linear functionals of the same matrix \(W_j\), formed from two matrix
exponentials of the current partial sum.  Once
\(Y_j=D_G^{-1/2}PW_jPD_G^{-1/2}\) is available, the score of a matching is
just a sum of entries of \(Y_j\) over its edges.  Since the matchings
partition the edges, all scores can be obtained in one edge scan.

We implement computing these scores in two different ways.  The dense
implementation forms the inverse square root and matrix exponentials
explicitly using approximate diagonalization and fast matrix
multiplication.  Its principal cost depends on the number of vertices
rather than the number of edges.  The sparse implementation represents
the inverse square root and exponentials by low-degree polynomials and
applies them through sparse matrix--vector products.  Its cost is roughly
the product of the numbers of vertices and edges.

\paragraph{Removing regularity and expansion assumptions.}
Two reductions turn the approximately regular expander algorithm into a general sparsification
algorithm.  First, a high-degree vertex is split into a fiber of
constant-degree copies, and the copies are connected by an auxiliary
expander.  The resulting graph has degrees within a constant factor of a
common scale and retains the original conductance up to constants.
Contracting every fiber maps a sparsifier of the regularized graph back to
a sparsifier of the original graph.

Second, we apply a deterministic expander
decomposition~\cite{Chuzhoy23}. We peel off low-degree vertices, keeping a constant fraction of the edges, then apply the decomposition and sparsify each expander.
A constant fraction of the current edges lies
inside the resulting expanders by the decomposition guarantee on inter-cluster edges and is sparsified in one round.  The
residual edge count decreases geometrically, and the cluster graphs from
all rounds form an edge-disjoint decomposition of the input.  Their local
spectral approximations therefore add to a global one.

For a weighted graph, we first separate the edges into dyadic weight
classes.  The sparse running-time terms sum over the classes because the
classes are edge-disjoint; only the output size and the
vertex-dependent dense work depend on the number of classes. We keep these small by aggregating edges.

\paragraph{Recursive self-improvement.}
The two implementations have complementary strengths.  The dense
implementation is insensitive to the number of edges but costs
\(n^\omega\), whereas the sparse implementation has a factor proportional
to the number of input edges, which could be quadratic.  The self-reduction uses the dense routine
locally to reduce the number of edges and then uses the sparse routine
globally.  Repeating this combination decreases the vertex exponent to
the barrier \(2\).

Suppose that we already have an algorithm whose vertex-dependent running
time is \(N^{\alpha_r}\), up to common parameter and polylogarithmic
factors.  Partition \(V(G)\) into blocks of size at most \(s\).  Every edge
belongs to exactly one graph supported on a pair of blocks, so these
block-pair graphs form an edge-disjoint decomposition of \(G\).  There are
\(O(n^2/s^2)\) such graphs, each on at most \(2s\) vertices.  Sparsifying
them recursively therefore costs
\[
    O\!\left(\frac{n^2}{s^2}s^{\alpha_r}\right)
    =
    O\!\left(n^2s^{\alpha_r-2}\right).
\]
Each block-pair sparsifier has only \(O(s)\) edges, up to the common factors, so their
union has only
\[
    O\!\left(\frac{n^2}{s}\right)
\]
edges while still approximating the entire graph.  The sparse
implementation can consequently compress this union in
\[
    O\!\left(n\cdot\frac{n^2}{s}\right)
    =
    O\!\left(\frac{n^3}{s}\right)
\]
time, again omitting the common factors.

Balancing the two costs gives
\[
    s=n^{1/(\alpha_r-1)}
    \qquad\text{and}\qquad
    \alpha_{r+1}
    =
    3-\frac{1}{\alpha_r-1}.
\]
If \(\delta_r:=\alpha_r-2\) denotes the excess over the quadratic
exponent, then
\[
    \delta_{r+1}=\frac{\delta_r}{1+\delta_r},
    \qquad
    \frac1{\delta_{r+1}}=\frac1{\delta_r}+1.
\]
Thus the reciprocal gap from \(2\) increases by one at every level, starting with $\alpha_0 = \omega$ we have
\[
    \alpha_r
    =
    2+\frac{1}{r+1/(\omega-2)}.
\]
After \(r=\Theta(\log n)\) levels, the remaining gap is
\(O(1/\log n)\), so
\[
    n^{\alpha_r}
    =
    n^2n^{O(1/\log n)}
    =
    O(n^2).
\]
We use accuracy \(\Theta(\eps/\log n)\) at each level, so the accumulated
error is at most \(\eps\); this changes only the 
polylogarithmic and \(\eps\)-dependent factors.

\paragraph{Organization.}
\Cref{sec:preliminaries} introduces the notation and basic facts used throughout.  \Cref{sec:matrix-chernoff} develops the
deterministic matrix-Chernoff selection procedure.
\Cref{sec:expander-sparsification} reduces expander sparsification to that
selection problem and gives the dense and sparse score implementations.
\Cref{sec:global-sparsification} combines degree regularization with a
deterministic expander decomposition and extends the construction to
weighted graphs.  Finally, \cref{sec:self-improvement} develops the
recursive combination and proves the main theorem.

\section{Preliminaries}
\label{sec:preliminaries}

All graphs are finite undirected weighted multigraphs.  Parallel edges are
allowed. We distinguish weighted degree from combinatorial degree throughout.

For $G=(V,E,w)$, let $\deg_G(v)$ be the weighted degree of $v$, including
self-loops, let $D_G$ be the diagonal degree matrix, and let
$\vol_G(S):=\sum_{v\in S}\deg_G(v)$. We also write $\deg^\#_G(v)$ and $\vol_G^\#(S)$ for unweighted (combinatorial) degree and volume, respectively. We count self-loops in weighted degree and volume but ignore them for combinatorial degree and volume. We count parallel non-loop edges with multiplicity for combinatorial degree. Let $\delta_G := \min_{v \in V} \deg_G(v)$ and $\delta^\#_G := \min_{v \in V} \deg^\#_G(v)$; let $\Delta_G$ and $\Delta^\#_G$ be maximum degree and maximum combinatorial degree. We write $E^\times(G)$ to denote all
non-loop edges, and for disjoint $A,B\subseteq V$ we write
$E_G(A,B)$ for the multiset of edges with one endpoint in $A$ and the
other in $B$.  For an edge multiset $F$, let
$w(F):=\sum_{e\in F}w(e)$.  The induced subgraph on $S$ is denoted by
$G[S]$, and the Laplacian is $L_G:=D_G-A_G$.
The normalized Laplacian is
$\cL_G:=D_G^{-1/2}L_GD_G^{-1/2}$.  We write $I_U$ for the
orthogonal projector onto a subspace $U$, viewed as the identity on $U$
and zero on $U^\perp$.
We discard isolated vertices.

The conductance of $G$ is
\[
    \Phi(G):=
    \min_{\emptyset\ne S\subsetneq V}
    \frac{w(E_G(S,V\setminus S))}
         {\min\{\vol_G(S),\vol_G(V\setminus S)\}}.
\]

For symmetric matrices, $A\preceq B$ means that $B-A$ is positive
semidefinite, $\norm{A}$ is the operator norm, and
$\ip{A}{B}:=\tr(AB)$. The following lemma is an easy consequence of Cheeger's inequality.

\begin{lemma}[Cheeger's inequality]
\label{lem:cheeger}
If $\Phi(G)\ge\phi>0$, then $G$ is connected,
$\ker\cL_G=\operatorname{span}(D_G^{1/2}\1)$, and, on
$\cU_G:=(D_G^{1/2}\1)^\perp$,
\[
    \frac{\phi^2}{2}I_{\cU_G}
    \preceq
    \cL_G\big|_{\cU_G}
    \preceq
    2I_{\cU_G}.
\]
\end{lemma}

\begin{proof}
Since $\Phi(G)>0$, the graph is connected.  Hence
$\ker L_G=\operatorname{span}(\1)$, and conjugation by
$D_G^{-1/2}$ gives
$\ker\cL_G=\operatorname{span}(D_G^{1/2}\1)$.
Cheeger's inequality gives
$\lambda_2(\cL_G)\ge\Phi(G)^2/2\ge\phi^2/2$, while every eigenvalue of a
normalized Laplacian is at most $2$.  Restricting to
$\cU_G=(D_G^{1/2}\1)^\perp$ proves the claim.
\end{proof}

We write $\omega$ for the exponent of square matrix multiplication.

\begin{definition}[Spectral sparsifier]
A weighted graph $H$ on $V(G)$ is a $(1\pm\eps)$-spectral sparsifier of
$G$ if
$(1-\eps)L_G\preceq L_H\preceq(1+\eps)L_G$.
\end{definition}

In every theorem that produces a spectral sparsifier we assume $\eps^{-1} \le n$ as otherwise the whole graph after aggregating parallel edges suffices.

\begin{lemma}[Edge-disjoint composition]
\label{lem:composition}
Suppose $G=G_1\mathbin{\dot\cup}\cdots\mathbin{\dot\cup}G_t$ is an
edge-disjoint union and $H_i$ is a $(1\pm\eps)$-spectral sparsifier of
$G_i$ for every $i\in[t]$.  Then
$H:=H_1\mathbin{\dot\cup}\cdots\mathbin{\dot\cup}H_t$ is a
$(1\pm\eps)$-spectral sparsifier of $G$.
\end{lemma}

\begin{proof}
Laplacians add over edge-disjoint unions.  Summing the approximation
inequalities for the $G_i$ proves the result.
\end{proof}

\section{Deterministic Matrix Chernoff}
\label{sec:matrix-chernoff}

This section derandomizes a matrix Chernoff bound which we use for the construction.  Given a positive semidefinite
family whose average is the identity and whose members are uniformly
bounded, we deterministically select a small multiset whose average still
approximates the identity.  The proof is the method of conditional
expectation applied to the matrix-Chernoff pessimistic estimator of
Wigderson and Xiao~\cite{WX08}.  We first give the exact selector and then
show that small additive error in its candidate scores is
sufficient.

\begin{lemma}[Chord bounds for the exponential]
\label{lem:chord-bounds}
If $0\preceq X\preceq I$, then, for every $\theta\in\R$,
\[
    e^{\theta X}
    \preceq
    I+(e^\theta-1)X.
\]
\end{lemma}

\begin{proof}
For $x\in[0,1]$, convexity of $e^{\theta x}$ gives
$e^{\theta x}\le(1-x)+xe^\theta$.  Applying this scalar inequality to
every eigenvalue of $X$ proves the matrix inequality by spectral
functional calculus.
\end{proof}

\begin{algorithm}[H]
\caption{\(\textsc{MatrixChernoffSelect}(A_1,\ldots,A_q,R,\eps,k)\)}
\label{alg:matrix-chernoff-selection}
\begin{algorithmic}[1]
\Require Positive semidefinite \(r\times r\) matrices satisfying
\(\frac1q\sum_iA_i=I\) and \(A_i\preceq RI\), where \(R\ge1\);
\(0<\eps\le1/2\); and \(k\ge4R\eps^{-2}\ln(4r)\)
\Ensure A sequence \(i_1,\ldots,i_k\in[q]\)
\State \(X_i\gets A_i/R\) for every \(i\in[q]\)
\State \(\theta_+\gets\log(1+\eps)\), \(\theta_-\gets-\log(1-\eps)\)
\State \(\rho_+\gets1+\eps/R\), \(\rho_-\gets1-\eps/R\), and \(S\gets0\)
\For{\(j=1,\ldots,k\)}
    \State \(u_j\gets
    e^{-\theta_+(1+\eps)k/R}\rho_+^{\,k-j}\)
    \State \(\ell_j\gets
    e^{\theta_-(1-\eps)k/R}\rho_-^{\,k-j}\)
    \State \(E_+\gets e^{\theta_+S}\) and \(E_-\gets e^{-\theta_-S}\)
    \State Choose \(i_j\) minimizing
    \(u_j\langle E_+,A_i\rangle-\ell_j\langle E_-,A_i\rangle\)
    over \(i\in[q]\)
    \State \(S\gets S+X_{i_j}\)
\EndFor
\State \Return \(i_1,\ldots,i_k\)
\end{algorithmic}
\end{algorithm}

\begin{lemma}[Deterministic matrix Chernoff]
\label{lem:deterministic-chernoff}
Let $A_1,\ldots,A_q$ be positive semidefinite $r\times r$ matrices such
that
\[
    \frac1q\sum_{i=1}^qA_i=I,
    \qquad
    0\preceq A_i\preceq RI
\]
for some $R\ge1$.  Let $0<\eps\le1/2$, and let
\begin{equation}
\label{eq:chernoff-k}
    k
    \ge
    4R\eps^{-2}\ln(4r).
\end{equation}
Then there are indices $i_1,\ldots,i_k\in[q]$, with repetition, such that
\[
    (1-\eps)I
    \preceq
    \frac1k\sum_{j=1}^kA_{i_j}
    \preceq
    (1+\eps)I.
\]
\cref{alg:matrix-chernoff-selection} finds such indices
deterministically.  Each iteration evaluates two matrix exponentials and
one linear score per candidate.
\end{lemma}

\begin{proof}
We first define a two-sided
trace-exponential estimator and show that it dominates the conditional
probability of failure.  We then show that its initial value is less than
one.  Finally, we show that the average one-step upper bound over all
candidates equals the current estimator, so conditional expectation
always supplies a nonincreasing choice.

Set
\[
    X_i:=A_i/R,
    \qquad
    \frac1q\sum_{i=1}^qX_i=\frac1R I,
    \qquad
    0\preceq X_i\preceq I.
\]
Define
\[
\begin{aligned}
    \theta_+&:=\ln(1+\eps),
    &\rho_+&:=1+\frac{\eps}{R},\\
    \theta_-&:=-\ln(1-\eps),
    &\rho_-&:=1-\frac{\eps}{R}.
\end{aligned}
\]
After fixing $i_1,\ldots,i_j$, write
\[
    S_j:=\sum_{h=1}^jX_{i_h}
\]
and define
\begin{equation}
\label{eq:psd-estimator}
\begin{aligned}
    \Phi_j
    :={}&
    e^{-\theta_+(1+\eps)k/R}
    \rho_+^{k-j}\tr e^{\theta_+S_j}\\
    &+
    e^{\theta_-(1-\eps)k/R}
    \rho_-^{k-j}\tr e^{-\theta_-S_j}.
\end{aligned}
\end{equation}

\paragraph{The conditional failure bound.}
We first show that $\Phi_j$ dominates the conditional failure
probability.  Condition on $i_1,\ldots,i_j$, and let
$J_{j+1},\ldots,J_k$ be independent uniform elements of $[q]$.  By
\cref{lem:chord-bounds},
\[
    \mathbb E e^{\theta_+X_J}
    \preceq
    I+\eps\mathbb E X_J
    =\rho_+I
\]
and
\[
    \mathbb E e^{-\theta_-X_J}
    \preceq
    I-\eps\mathbb E X_J
    =\rho_-I.
\]
Repeated applications of the Golden--Thompson inequality give
\[
\begin{aligned}
    \mathbb E\tr\exp\!\left(
        \theta_+S_j+
        \theta_+\sum_{h=j+1}^kX_{J_h}
    \right)
    &\le
    \rho_+^{k-j}\tr e^{\theta_+S_j},\\
    \mathbb E\tr\exp\!\left(
        -\theta_-S_j-
        \theta_-\sum_{h=j+1}^kX_{J_h}
    \right)
    &\le
    \rho_-^{k-j}\tr e^{-\theta_-S_j}.
\end{aligned}
\]
For the upper tail,
\begin{align*}
&\Pr\!\left[
    \lambda_{\max}\!\left(
        S_j+\sum_{h=j+1}^kX_{J_h}
    \right)>(1+\eps)\frac{k}{R}
    \,\middle|\,
    i_1,\ldots,i_j
\right]\\
&\quad\le
\Pr\!\left[
    \tr\exp\!\left(
        \theta_+S_j+
        \theta_+\sum_{h=j+1}^kX_{J_h}
    \right)
    >e^{\theta_+(1+\eps)k/R}
    \,\middle|\,
    i_1,\ldots,i_j
\right]\\
&\quad\le
    e^{-\theta_+(1+\eps)k/R}
    \rho_+^{k-j}\tr e^{\theta_+S_j}.
\end{align*}
The first inequality uses
$\tr e^Y\ge e^{\lambda_{\max}(Y)}$; the second is Markov's
inequality.  Similarly,
\begin{align*}
&\Pr\!\left[
    \lambda_{\min}\!\left(
        S_j+\sum_{h=j+1}^kX_{J_h}
    \right)<(1-\eps)\frac{k}{R}
    \,\middle|\,
    i_1,\ldots,i_j
\right]\\
&\quad\le
\Pr\!\left[
    \tr\exp\!\left(
        -\theta_-S_j-
        \theta_-\sum_{h=j+1}^kX_{J_h}
    \right)
    >e^{-\theta_-(1-\eps)k/R}
    \,\middle|\,
    i_1,\ldots,i_j
\right]\\
&\quad\le
    e^{\theta_-(1-\eps)k/R}
    \rho_-^{k-j}\tr e^{-\theta_-S_j}.
\end{align*}
A union bound proves that $\Phi_j$ dominates the conditional probability
of either failure.

\paragraph{The initial value.}
We next bound $\Phi_0$.  Since $\ln(1+x)\le x$,
\[
\begin{aligned}
    e^{-\theta_+(1+\eps)k/R}\rho_+^k
    &\le
    \exp\!\left(
        -\frac{k}{R}
        \bigl((1+\eps)\ln(1+\eps)-\eps\bigr)
    \right),\\
    e^{\theta_-(1-\eps)k/R}\rho_-^k
    &\le
    \exp\!\left(
        -\frac{k}{R}
        \bigl(\eps+(1-\eps)\ln(1-\eps)\bigr)
    \right).
\end{aligned}
\]
For $0\le\eps\le1/2$, both quantities in parentheses are at least
$\eps^2/3$: their values and first derivatives vanish at zero, while their
second derivatives are at least $2/3$ and $1$, respectively.  Hence
\[
    \Phi_0
    \le
    2r\exp\!\left(-\frac{k\eps^2}{3R}\right)
    \le\frac12
\]
by \eqref{eq:chernoff-k}.

\paragraph{Decrease of the estimator.}
It remains to prove that the choice in
\cref{alg:matrix-chernoff-selection} does not increase the
estimator.
Suppose $i_1,\ldots,i_{j-1}$ have been fixed and write $S=S_{j-1}$.
For a candidate $i$, Golden--Thompson and
\cref{lem:chord-bounds} give
\[
\begin{aligned}
    &\tr e^{\theta_+(S+X_i)}
    \le
    \tr\left(e^{\theta_+S}e^{\theta_+X_i}\right)
    \le
    \tr e^{\theta_+S}
    +\eps\ip{e^{\theta_+S}}{X_i},\\
    &\tr e^{-\theta_-(S+X_i)}
    \le
    \tr e^{-\theta_-S}
    -\eps\ip{e^{-\theta_-S}}{X_i}.
\end{aligned}
\]
Set
\[
\begin{aligned}
    u_j:=e^{-\theta_+(1+\eps)k/R}\rho_+^{k-j}, \qquad
    \ell_j:=e^{\theta_-(1-\eps)k/R}\rho_-^{k-j}.
\end{aligned}
\]
The candidate-dependent part of the resulting upper bound on $\Phi_j$ is
\begin{equation}
\label{eq:chernoff-score}
    \eps\left(
        u_j\ip{e^{\theta_+S}}{X_i}
        -
        \ell_j\ip{e^{-\theta_-S}}{X_i}
    \right).
\end{equation}
Averaging this upper bound over $i$ and using
$\mathbb E X_i=I/R$ gives exactly $\Phi_{j-1}$, because
$u_{j-1}=\rho_+u_j$ and $\ell_{j-1}=\rho_-\ell_j$.
Therefore some candidate makes $\Phi_j\le\Phi_{j-1}$; choose one such
candidate.

Inductively, $\Phi_k\le\Phi_0<1$.  Since $\Phi_k$ dominates the
indicator of either deterministic failure, neither failure occurs.  Thus
\[
    (1-\eps)\frac{k}{R}I
    \preceq
    S_k
    \preceq
    (1+\eps)\frac{k}{R}I.
\]
Multiplying by $R/k$ and substituting $X_i=A_i/R$ proves the claim.
\end{proof}

\begin{lemma}[Inexact candidate scores]
\label{lem:inexact-scores}
In the setting of \cref{lem:deterministic-chernoff}, suppose that at every
iteration the candidate scores
\[
    u_j\ip{e^{\theta_+S}}{A_i}
    -
    \ell_j\ip{e^{-\theta_-S}}{A_i}
\]
are computed for all candidates with additive error at most
\[
    \zeta:=\frac{R}{16\eps k}.
\]
Choosing a candidate of minimum approximate value still yields
\[
    (1-\eps)I
    \preceq
    \frac1k\sum_{j=1}^kA_{i_j}
    \preceq
    (1+\eps)I.
\]
\end{lemma}

\begin{proof}
Fix an iteration and let $\sigma_i$ and $\widetilde\sigma_i$ denote the
exact and approximate scores of candidate $i$, respectively.  If
$i^\star$ minimizes the approximate score, then
\[
    \sigma_{i^\star}
    \le \widetilde\sigma_{i^\star}+\zeta
    \le \frac1q\sum_{i=1}^q\widetilde\sigma_i+\zeta
    \le \frac1q\sum_{i=1}^q\sigma_i+2\zeta.
\]
Because $q^{-1}\sum_iA_i=I$, the average exact score is precisely the
score obtained by replacing $A_i$ with $I$.

In the one-step estimator bound from
\eqref{eq:chernoff-score}, scores written in terms of $A_i$ are
multiplied by $\eps/R$.  The selected candidate can therefore increase
the estimator by at most
\[
    \frac{\eps}{R}\cdot2\zeta
    \le \frac1{8k}.
\]
Consequently, after $j$ inexact choices,
\[
    \Phi_j
    \le \Phi_0+\frac{j}{8k}
    \le \frac12+\frac18
    =\frac58<1.
\]
At $j=k$, either failure would force the corresponding term of
$\Phi_k$ to be at least one, a contradiction.  Hence neither failure
occurs.
\end{proof}

\section{Sparsifying an Expander}
\label{sec:expander-sparsification}

We now convert an expanding graph into an instance of the selection
problem from the preceding section.  Conjugation by the inverse square
root of the normalized Laplacian makes the normalized edge Laplacians sum
to the identity.  A decomposition into matchings groups these edges into 
few matrices, while expansion and degree regularity ensure that no group
has large operator norm.  Deterministic matrix Chernoff can then select
only a small number of matchings.  After establishing this reduction, we
give two implementations of its scores: a dense implementation based on
matrix multiplication and a sparse implementation based on 
matrix--vector products to evaluate polynomials approximating the relevant matrix functions.

Let $G$ be a connected weighted
multigraph, possibly with self-loops, with $\Phi(G)\ge\phi$, suppose every
non-loop edge has weight at most $\beta$, and partition the non-loop edges
into matchings
$E^\times(G)=M_1\mathbin{\dot\cup}\cdots\mathbin{\dot\cup}M_q$.
Write $\cU:=(D_G^{1/2}\1)^\perp$, and let
$P:=(\cL|_{\cU})^{-1/2}$, extended by zero on
$\operatorname{span}(D_G^{1/2}\1)$.  Define
\begin{equation}
\label{eq:Ai}
    A_i
    :=
    qP D_G^{-1/2}L_{M_i}D_G^{-1/2}P.
\end{equation}
We regard the \(A_i\)'s as operators on \(\cU\).  Thus the identity,
trace, and matrix functions appearing in
\cref{alg:matrix-chernoff-selection} are taken on this
\((n-1)\)-dimensional subspace.

Then
\[
    \frac1q\sum_{i=1}^qA_i=I_\cU.
\]
Since
$L_{M_i}\preceq2\beta I$, \cref{lem:cheeger} gives
\begin{equation}
\label{eq:R-general}
    0\preceq A_i\preceq RI_\cU,
    \qquad
    R:=\frac{4\beta q}{\delta \phi^2}.
\end{equation}

Set
\begin{equation}
\label{eq:k-matching}
    k
    :=
    \left\lceil
        4R\eps^{-2}\ln(4n)
    \right\rceil
\end{equation}
as required by \cref{lem:deterministic-chernoff}.  The resulting
selection procedure is summarized in
\cref{alg:matching-sparsifier}.

\begin{algorithm}[H]
\caption{\(\textsc{MatchingSparsifier}
(G,M_1,\ldots,M_q,\beta,\phi,\eps)\)}
\label{alg:matching-sparsifier}
\begin{algorithmic}[1]
\Require A connected graph \(G\) with \(\Phi(G)\ge\phi\), non-loop
weights at most \(\beta\), and a matching partition
\(M_1,\ldots,M_q\) of \(E^\times(G)\)
\Ensure A \((1\pm\eps)\)-spectral sparsifier of \(G\)
\State Form \(\cL\), \(\cU\), \(P\), and \(A_1,\ldots,A_q\) as in
\eqref{eq:Ai}
\State \(R\gets4\beta q/(\delta \phi^2)\) and
\(k\gets\lceil4R\eps^{-2}\ln(4n)\rceil\)
\State \((i_1,\ldots,i_k)\gets
\Call{MatrixChernoffSelect}{A_1,\ldots,A_q,R,\eps,k}\)
\State \Return the weighted multigraph
\(\frac{q}{k}(M_{i_1}\mathbin{\dot\cup}\cdots
\mathbin{\dot\cup}M_{i_k})\)
\end{algorithmic}
\end{algorithm}

By \cref{lem:deterministic-chernoff}, the selected indices satisfy
\[
    (1-\eps)I_\cU
    \preceq
    \frac1k\sum_{j=1}^kA_{i_j}
    \preceq
    (1+\eps)I_\cU.
\]
Conjugating by $\cL^{1/2}|_\cU$ and then by $D_G^{1/2}$ shows that the
output
\begin{equation}
\label{eq:matching-output}
    H
    :=
    \frac{q}{k}
    \left(
        M_{i_1}\mathbin{\dot\cup}\cdots\mathbin{\dot\cup}M_{i_k}
    \right)
\end{equation}
is a $(1\pm\eps)$-spectral sparsifier of $G$.  Since every selected
matching has at most $n/2$ edges,
\begin{equation}
\label{eq:expander-size}
    |E(H)|
    =
    O\!\left(nR\eps^{-2}\log(n)\right).
\end{equation}

It remains to implement the score evaluations in
\cref{alg:matrix-chernoff-selection}.  We give a dense
implementation based on explicitly computing the matrices and a sparse
implementation based on polynomial approximations.

\subsection{Dense Score Evaluation}
We first record the score normalization and error criterion shared by the
dense and sparse implementations, and then give the dense implementation.
The sparse subsection will reuse these estimates and only needs to show
how polynomial evaluation attains the required componentwise accuracy.

Throughout this subsection, \(S,U_j,V_j,Q_{+,j}\), and \(Q_{-,j}\) are
operators on \(\cU\); accordingly, \(I\), \(\tr\), and the matrix
exponential refer to their restrictions to \(\cU\).  We identify these
operators with their zero extensions to \(\cU^\perp\) whenever they are
sandwiched by \(P\).

At iteration $j$, let $S$ be the sum of the previously selected matrices
$X_i=A_i/R$.  Define
\[
    W_j
    :=
    u_j e^{\theta_+S}
    -
    \ell_j e^{-\theta_-S},
\]
where $u_j$ and $\ell_j$ are as in the proof of
\cref{lem:deterministic-chernoff}.  The candidate-independent terms may be
omitted, so it is enough to evaluate $\ip{W_j}{A_i}$.

Assume inductively that the preceding approximate choices met the approximate score
tolerance in \cref{lem:inexact-scores}.  Set
\[
\begin{aligned}
    U_j&:=u_{j-1}e^{\theta_+S},&
    V_j&:=\ell_{j-1}e^{-\theta_-S}.
\end{aligned}
\]
By \cref{lem:inexact-scores},
\[
    \tr U_j+\tr V_j=\Phi_{j-1}<\frac34.
\]
Hence $U_j$ and $V_j$ are positive semidefinite and have operator norm at
most one.  Since $u_j=\rho_+^{-1}u_{j-1}$ and
$\ell_j=\rho_-^{-1}\ell_{j-1}$,
\[
    W_j=\rho_+^{-1}U_j-\rho_-^{-1}V_j,
    \qquad
    \rho_+^{-1}\le1,
    \quad
    \rho_-^{-1}\le2.
\]
In particular, both exponentials can be evaluated by shifting by an identity matrix.  Define
\begin{equation}
\label{eq:Qpm}
    Q_{+,j}:=-(\log u_{j-1})I-\theta_+S,
    \qquad
    Q_{-,j}:=-(\log\ell_{j-1})I+\theta_-S.
\end{equation}
Then
\[
    U_j=e^{-Q_{+,j}},
    \qquad
    V_j=e^{-Q_{-,j}}.
\]
The norm bounds on \(U_j\) and \(V_j\) imply that
\(Q_{+,j},Q_{-,j}\succeq0\).  First,
\[
    Q_{+,j}
    \preceq
    -(\log u_{j-1})I
    \preceq
    \frac{2\eps k}{R}I.
\]
For the other exponential, \(\tr U_j<1\) implies
\(U_j\prec I\).  Since \(U_j=u_{j-1}e^{\theta_+S}\), it follows that
\[
    S
    \preceq
    \frac{-\log u_{j-1}}{\theta_+}I
    \preceq
    \frac{(1+\eps)k}{R}I.
\]
Consequently,
\[
\begin{aligned}
    Q_{-,j}
    =
    -(\log\ell_{j-1})I+\theta_-S
    \preceq
    \left(
        -\log\ell_{j-1}
        +\frac{\theta_-}{\theta_+}(-\log u_{j-1})
    \right)I.
\end{aligned}
\]
From the definitions we have
\[
\begin{aligned}
    -\log u_{j-1}
    &=
    \frac{\theta_+(1+\eps)k}{R}
    -(k-j+1)\log\rho_+,\\
    -\log\ell_{j-1}
    &=
    -\frac{\theta_-(1-\eps)k}{R}
    -(k-j+1)\log\rho_-.
\end{aligned}
\]
Thus,
\[
\begin{aligned}
    &-\log\ell_{j-1}
    +\frac{\theta_-}{\theta_+}(-\log u_{j-1})\\
    &\qquad=
    \frac{2\eps\theta_-k}{R}
    +(k-j+1)
    \left(
        -\log\rho_-
        -\frac{\theta_-}{\theta_+}\log\rho_+
    \right)\\
    &\qquad\le
    \frac{2\eps\theta_-k}{R}
    +(k-j+1)(-\log\rho_-).
\end{aligned}
\]
The inequality uses \(\rho_+>1\).  Since \(0<\eps\le1/2\) and \(R\ge1\),
we have \(\theta_-=-\log(1-\eps)\le\log2<1\) and
\[
    -\log\rho_-
    =
    -\log\!\left(1-\frac{\eps}{R}\right)
    \le\frac{2\eps}{R}.
\]
Therefore
\[
    Q_{-,j}
    \preceq
    \left(
        \frac{2\eps k}{R}
        +\frac{2\eps(k-j+1)}{R}
    \right)I
    \preceq
    \frac{4\eps k}{R}I.
\]
Thus
\begin{equation}
\label{eq:Q-range}
    0\preceq Q_{+,j}\preceq\frac{2\eps k}{R}I,
    \qquad
    0\preceq Q_{-,j}\preceq\frac{4\eps k}{R}I.
\end{equation}

For the full-space implementations, let \(\Pi:=I_{\cU}\), extend \(S\)
by zero on \(\cU^\perp\), and define
\[
    Q_{+,j}^{\circ}
    :=
    -(\log u_{j-1})\Pi-\theta_+S,
    \qquad
    Q_{-,j}^{\circ}
    :=
    -(\log\ell_{j-1})\Pi+\theta_-S.
\]
These matrices vanish on \(\cU^\perp\) and restrict to \(Q_{+,j}\) and
\(Q_{-,j}\) on \(\cU\).

We will use the following elementary bound in both
implementations.

\begin{lemma}[Perturbation of a negative exponential]
\label{lem:exp-perturbation}
If $A,B\succeq0$, then
\[
    \norm{e^{-A}-e^{-B}}\le\norm{A-B}.
\]
\end{lemma}

\begin{proof}
Set \(F(t):=e^{-(1-t)A}e^{-tB}\).  Differentiating and integrating from
\(0\) to \(1\) gives
\[
    e^{-A}-e^{-B}
    =-
    \int_0^1
    e^{-(1-t)A}(A-B)e^{-tB}\,dt.
\]
Both exponential factors have operator norm at most one, so taking norms
proves the claim.
\end{proof}

\paragraph{The score matrix.}
Let
\[
    \overline Y_j:=PW_jP,
    \qquad
    Y_j:=D_G^{-1/2}\overline Y_jD_G^{-1/2}.
\]
Then
\begin{equation}
\label{eq:dense-score-edge-scan}
\begin{aligned}
    \ip{W_j}{A_i}
    &=q\ip{Y_j}{L_{M_i}}\\
    &=q\sum_{\{u,v\}\in M_i}
        w(uv)\left(
            (Y_j)_{uu}+(Y_j)_{vv}-2(Y_j)_{uv}
        \right).
\end{aligned}
\end{equation}
Thus all candidate scores are obtained by one scan over the edges after
$Y_j$ has been formed.

It will therefore suffice in either implementation to approximate
$\overline Y_j$ in operator norm.  If
$\norm{\widetilde{\overline Y}_j-\overline Y_j}\le\tau_Y$, then for every
matching $M_i$,
\begin{equation}
\label{eq:score-error-reduction}
\begin{aligned}
q\left|
 \ip{\widetilde{\overline Y}_j-\overline Y_j}
      {D_G^{-1/2}L_{M_i}D_G^{-1/2}}
\right|
&\le
q\tau_Y\tr(D_G^{-1/2}L_{M_i}D_G^{-1/2})\\
&\le \frac{\beta qn}{\delta}\tau_Y
=\frac{Rn\phi^2}{4}\tau_Y
\le Rn\tau_Y.
\end{aligned}
\end{equation}
The second inequality uses that $M_i$ is a matching whose edge weights
are at most $\beta$, and the equality uses
$R=4\beta q/(\delta \phi^2)$.  Consequently,
\begin{equation}
\label{eq:tau-Y}
    \tau_Y:=\frac{1}{64n\eps k}
\end{equation}
is sufficient for the tolerance in \cref{lem:inexact-scores}.

We also record how errors in the three factors of $\overline Y_j$
propagate.  Suppose
\[
    \norm{\widetilde P-P}\le\tau_P,
    \qquad
    \norm{\widetilde U_j-U_j},
    \norm{\widetilde V_j-V_j}\le\tau_E,
\]
where $\tau_P,\tau_E\le1$, and put
\[
    \widetilde W_j
    :=
    \rho_+^{-1}\widetilde U_j
    -
    \rho_-^{-1}\widetilde V_j,
    \qquad
    \widetilde{\overline Y}_j
    :=
    \widetilde P\widetilde W_j\widetilde P.
\]
Since $\norm P=O(1/\phi)$, $\norm{\widetilde P}=O(1/\phi)$,
$\norm{W_j}\le3$, and
$\norm{\widetilde W_j-W_j}\le3\tau_E$, a telescoping expansion gives
\begin{equation}
\label{eq:Y-assembly-error}
    \norm{\widetilde{\overline Y}_j-\overline Y_j}
    =
    O\!\left(
        \frac{\tau_P}{\phi}
        +
        \frac{\tau_E}{\phi^2}
    \right).
\end{equation}

\paragraph{Approximating the partial sum.}
Let $K$ be the multiset union of the previously selected matchings and
put $X:=D_G^{-1/2}L_KD_G^{-1/2}$.  Then
\begin{equation}
\label{eq:S-representation}
    S
    =
    \frac{q}{R}PXP.
\end{equation}
Every edge of $K$ is a non-loop edge of $G$, and its multiplicity can be
stored with that edge of $G$.  Hence $L_Kx$ can be computed in $O(m)$
time.  Moreover, because $K$ is the union of $j-1$ matchings,
\[
    \norm X\le\frac{2\beta(j-1)}{\delta}.
\]
If $\widetilde S:=(q/R)\widetilde P X\widetilde P$, then
\[
    PXP-\widetilde P X\widetilde P
    =
    (P-\widetilde P)XP
    +
    \widetilde P X(P-\widetilde P).
\]
Using $q/R=\delta \phi^2/(4\beta)$ and
$\norm P,\norm{\widetilde P}=O(1/\phi)$ therefore gives
\begin{equation}
\label{eq:S-perturbation}
    \norm{S-\widetilde S}
    =
    O((j-1)\phi\tau_P).
\end{equation}

\paragraph{Dense implementation and accuracy.}
For a graph with $n$ vertices and $m$ edges, define
\begin{equation}
\label{eq:L-definition}
    L
    :=
    \max\left\{
        2,
        \log\!\left(\frac{64(n+m)^4}{\eps\phi}\right)
    \right\}.
\end{equation}
By Corollary~2.4 of~\cite{Sob25}, a real symmetric \(n\times n\)
matrix can be diagonalized deterministically to operator-norm error
\(\eta_{\mathrm{diag}}\) in
\(O(n^\omega\log n+n^2\log^{O(1)}(n/\eta_{\mathrm{diag}}))\)
time.

Write $A:=\cL|_{\cU}$ and apply this algorithm to $A$ using
$\eta_{\mathrm{diag}}=e^{-CL}$ for a sufficiently large absolute
constant $C$. We thus obtain $\widehat A$, an approximate diagonalized version. 
We have
\[
    \norm{\widehat A-A}\le\eta_{\mathrm{diag}}
    \le\frac{\phi^2}{4}.
\]
The spectrum of \(A\) lies in \([\phi^2/2,2]\) by
\cref{lem:cheeger}.  Since
\(\widehat A-A\succeq-\eta_{\mathrm{diag}}I\), we have
\[
    \widehat A
    \succeq
    A-\eta_{\mathrm{diag}}I
    \succeq
    \left(\frac{\phi^2}{2}-\eta_{\mathrm{diag}}\right)I
    \succeq
    \frac{\phi^2}{4}I.
\]
Thus \(\widehat A\) is positive definite on \(\cU\).
Define
\[
    P:=A^{-1/2},
    \qquad
    \widetilde P:=\widehat A^{-1/2}
\]
on $\cU$, extending both operators by zero on $\cU^\perp$.

For every positive definite matrix $M$,
\begin{equation}
\label{eq:inverse-square-root-integral}
    M^{-1/2}
    =
    \frac1\pi
    \int_0^\infty
        t^{-1/2}(M+tI)^{-1}\,dt.
\end{equation}
By the spectral theorem it suffices to verify this identity for
a scalar $\lambda>0$, in which case the substitution $t=\lambda s$
gives
\[
    \frac1\pi\int_0^\infty
        \frac{t^{-1/2}}{\lambda+t}\,dt
    =
    \frac{\lambda^{-1/2}}{\pi}
    \int_0^\infty\frac{s^{-1/2}}{1+s}\,ds
    =
    \lambda^{-1/2}.
\]

Applying \eqref{eq:inverse-square-root-integral} to $A$ and
$\widehat A$, and using the identity
\[
    (\widehat A+tI)^{-1}-(A+tI)^{-1}
    =
    (\widehat A+tI)^{-1}
    (A-\widehat A)
    (A+tI)^{-1},
\]
we obtain
\[
\begin{aligned}
    \norm{\widetilde P-P}
    &\le
    \frac{\eta_{\mathrm{diag}}}{\pi}
    \int_0^\infty
        t^{-1/2}
        \norm{(\widehat A+tI)^{-1}}
        \norm{(A+tI)^{-1}}
        \,dt\\
    &\le
    \frac{\eta_{\mathrm{diag}}}{\pi}
    \int_0^\infty
        \frac{t^{-1/2}}
             {(t+\phi^2/4)^2}
        \,dt.
\end{aligned}
\]

Here we use that for every $c>0$,
\[
    \int_0^\infty\frac{t^{-1/2}}{(t+c)^2}\,dt
    =
    \frac{\pi}{2c^{3/2}}.
\]

Taking $c=\phi^2/4$ therefore yields
\[
    \norm{\widetilde P-P}
    \le
    \frac{4\eta_{\mathrm{diag}}}{\phi^3}
    =
    O(\phi^{-3}\eta_{\mathrm{diag}}).
\]

Thus, by taking
$C$ sufficiently large, we obtain
\[
    \norm{\widetilde P-P}=e^{-\Theta(L)}.
\]

The dense routine is used below only on the degree-regularized instances,
for which \(R=O(\phi^{-2})\).  Hence
\[
    k=O(\phi^{-2}\eps^{-2}\log n)
    \qquad\text{and}\qquad
    \log\!\left(\frac{nk}{\eps\phi}\right)=O(L).
\]
Thus polynomial factors in \(k\), \(1/\eps\), and \(1/\phi\) can be
absorbed by increasing the constant \(C\).

Form $\widetilde S$ from \eqref{eq:S-representation} using
$\widetilde P$.  By \eqref{eq:S-perturbation},
$\norm{\widetilde S-S}=e^{-\Theta(L)}$.  Let \(\gamma_j\) be the
explicit upper bound furnished by \eqref{eq:S-perturbation}, multiplied
by \(\max\{\theta_+,\theta_-\}\).  Then
\[
    \gamma_j=e^{-\Theta(L)}
    \qquad\text{and}\qquad
    \gamma_j\ge
    \max\{\theta_+,\theta_-\}\norm{\widetilde S-S}.
\]
Replace \(S\) by \(\widetilde S\) in the restrictions
\(Q_{+,j},Q_{-,j}\), and shift each resulting matrix upward by
\(\gamma_jI_{\cU}\).  The shifted matrices are positive semidefinite and
differ from the corresponding exact matrices by
\(e^{-\Theta(L)}\).  Diagonalize their restrictions to \(\cU\), replace
negative approximate eigenvalues by zero, apply \(x\mapsto e^{-x}\), and
extend the resulting operators by zero on \(\cU^\perp\).  This produces
\(\widetilde U_j,\widetilde V_j\succeq0\) satisfying
\[
    \norm{\widetilde U_j-U_j},
    \norm{\widetilde V_j-V_j}
    =
    e^{-\Theta(L)}
\]
by \cref{lem:exp-perturbation}.  Scalar evaluation errors of
\(e^{-CL}\) are absorbed in the same bound.

Finally form
\[
    \widetilde{\overline Y}_j
    =
    \widetilde P
    \left(
        \rho_+^{-1}\widetilde U_j
        -
        \rho_-^{-1}\widetilde V_j
    \right)
    \widetilde P.
\]
By \eqref{eq:Y-assembly-error}, increasing $C$ if necessary gives
$\norm{\widetilde{\overline Y}_j-\overline Y_j}\le\tau_Y$.
Equations \eqref{eq:score-error-reduction} and \eqref{eq:tau-Y} then give
the accuracy required by \cref{lem:inexact-scores}.

The approximate diagonalizations, scalar evaluations, and \(O(1)\)
additional matrix multiplications compute
$\widetilde{\overline Y}_j$ in
\begin{equation}
\label{eq:dense-expander-time}
    O\!\left(m + n^\omega\log n+n^2L^{O(1)}\right)
\end{equation}
time.  The final edge scan costs $O(m)$ and is included in this bound.

\subsection{Polynomial Approximation}
The sparse implementation obtains the componentwise accuracies used in
\eqref{eq:Y-assembly-error} without diagonalization.  It uniformly
approximates the inverse square root and negative exponential on their
relevant spectral intervals and applies the resulting polynomials using
matrix--vector products.  A degree-\(d\) polynomial in a Laplacian can be
applied to a vector using \(d\) sparse matrix--vector products.

We prove the inverse-square-root approximation in
\cref{app:inverse-square-root}; the exponential approximation is
\cite[Theorem~4.1]{SV14}.

\begin{restatable}[Inverse-square-root approximation]
    {lemma}{inversesqrtapproximation}
\label{lem:inverse-sqrt}
For every \(0<a\le1\) and \(0<\tau<1/2\), one can compute a polynomial
\(p\) of degree
\[
    \deg p
    =
    O\!\left(
        a^{-1/2}\log\!\left(\frac{1}{a\tau}\right)
    \right)
\]
satisfying
\[
    \sup_{x\in[a,2]}|p(x)-x^{-1/2}|\le\tau.
\]
\end{restatable}

\begin{lemma}[Approximation to the exponential]
\label{lem:exp-approx}
Whenever $0\preceq A\preceq RI$ and $0<\tau<1$, one can compute a polynomial $p$
of degree
\[
    \deg p
    =
    O\! \left(
        \sqrt{R\log\frac1\tau}
        +\log\frac1\tau
    \right)
\]
satisfying
\[
    e^{-A}-\tau I\preceq p(A)\preceq e^{-A}+\tau I.
\]
\end{lemma}

The corresponding matrix guarantees for scalar approximation follow immediately by applying the
polynomials to the eigenvalues.

\subsection{Sparse Score Evaluation}
\label{sec:sparse-evaluation}

The sparse implementation follows the same score identity as the dense
one but constructs the score matrix one column at a time.  Each
column requires applications of polynomials in sparse matrices and of two matrix
exponentials, and each polynomial step is a matrix--vector
multiplication.

The analysis is inductive over the selection iterations.  At iteration
$j$, use the matrices $S,U_j,V_j,Q_{+,j},Q_{-,j}$ and the multiset $K$
defined in the preceding subsection.  The induction hypothesis gives
bounds on $U_j$ and $V_j$, while \eqref{eq:Q-range} bounds the
intervals on which their negative exponentials must be approximated.
It remains to obtain sufficiently accurate polynomial representations of
$P$ and these two exponentials.

Let
\begin{equation}
\label{eq:Lambda-definition}
    \Lambda
    :=
    \max\left\{
        2,
        \log\!\left(
            \frac{256n k}{\eps\phi}
        \right)
    \right\}.
\end{equation}
Let \(\Pi\) be the orthogonal projection onto
\(\cU=(D_G^{1/2}\1)^\perp\).  Set
\begin{equation}
\label{eq:tau-p}
    \tau_P:=e^{-\Omega(\Lambda)}
\end{equation}
with a sufficiently large absolute constant in the exponent.  By
\cref{lem:cheeger}, the spectrum of \(\cL|_{\cU}\) lies in
\([\phi^2/2,2]\).  Apply
\cref{lem:inverse-sqrt} with \(a=\phi^2/2\) and accuracy \(\tau_P\), and
let \(p_P\) be the resulting polynomial.  Since
\(\log(1/(\phi^2\tau_P))=O(\Lambda)\), its degree satisfies
\begin{equation}
\label{eq:dP}
    d_P:=\deg p_P=O(\phi^{-1}\Lambda).
\end{equation}
Define the polynomial operator
\begin{equation}
\label{eq:Ptilde-definition}
    \widetilde P:=\Pi p_P(\cL)\Pi.
\end{equation}
The operator \(P\) equals
\((\cL|_{\cU})^{-1/2}\) on \(\cU\) and is zero on its orthogonal
complement.  Hence the guarantee of
\cref{lem:inverse-sqrt} gives
\[
    \norm{P-\widetilde P}\le\tau_P.
\]
One application of \(\widetilde P\) costs
\(O(md_P)\) time: the projections cost \(O(n)\), and
each multiplication by \(\cL\) costs \(O(m)\).

Define \(\widetilde S\) by replacing \(P\) with \(\widetilde P\) in
\eqref{eq:S-representation}.  Equation \eqref{eq:S-perturbation} gives
\[
    \norm{S-\widetilde S}
    =
    O((j-1)\phi\tau_P).
\]
Because \(j-1\le k\), \(\theta_+,\theta_-\le1\), and
\(\Lambda\ge\log(256nk/(\eps\phi))\), choosing the constant in
\eqref{eq:tau-p} sufficiently large gives a bound
\(\widehat\gamma_j\) satisfying
\[
    \widehat\gamma_j
    \ge
    \max\{\theta_+,\theta_-\}
    \norm{S-\widetilde S},
    \qquad
    \widehat\gamma_j=e^{-\Omega(\Lambda)}.
\]

Form \(Q_{+,j}^{\circ}\) and \(Q_{-,j}^{\circ}\) using
\(\widetilde S\), and add \(\widehat\gamma_j\Pi\) to each.  Denote the
resulting matrices by
\(\widehat Q_{+,j}^{\circ}\) and
\(\widehat Q_{-,j}^{\circ}\).  They are positive semidefinite and differ
from their exact counterparts by \(e^{-\Omega(\Lambda)}\).  After an adjustment of constants,
\[
    0\preceq
    \widehat Q_{\pm,j}^{\circ}
    \preceq TI,
    \qquad
    T:=\frac{5\eps k}{R}.
\]
By \cref{lem:exp-perturbation}, their projected negative exponentials
differ from \(U_j\) and \(V_j\) by \(e^{-\Omega(\Lambda)}\).
Applying \cref{lem:exp-approx} and projecting on both sides gives
polynomial operators
\[
    \widetilde U_j
    :=
    \Pi p_{+,j}(\widehat Q_{+,j}^{\circ})\Pi,
    \qquad
    \widetilde V_j
    :=
    \Pi p_{-,j}(\widehat Q_{-,j}^{\circ})\Pi
\]
with operator error \(e^{-C\Lambda}\) and degree
\begin{equation}
\label{eq:dE}
\begin{aligned}
    d_E
    =
    O\!\left(
        \sqrt{\frac{\eps k}{R}\Lambda}
        +\Lambda
    \right)
    =
    O\!\left(
        \sqrt{\eps^{-1}\log(n)\Lambda}
        +\Lambda
    \right).
\end{aligned}
\end{equation}
Each evaluation requires \(d_E\) applications of \(\widetilde S\).
By \eqref{eq:S-representation}, one such application uses two
applications of \(\widetilde P\) and one application of \(L_K\), and
therefore costs \(O(md_P)\).  Thus an exponential application costs
\(O(md_Pd_E)\).

\cref{alg:sparse-score-evaluation} summarizes one iteration.

\begin{algorithm}[H]
\caption{\(\textsc{SparseScores}\) for iteration \(j\)}
\label{alg:sparse-score-evaluation}
\begin{algorithmic}[1]
\Require \(G\), the matchings \(M_1,\ldots,M_q\), and the multiplicity
representation of
\(K=M_{i_1}\mathbin{\dot\cup}\cdots\mathbin{\dot\cup}M_{i_{j-1}}\)
\Ensure An approximate score \(\widetilde\sigma_i\) for every \(i\in[q]\)
\State Construct \(p_P\) using
\cref{lem:inverse-sqrt} with
\(a=\phi^2/2\) and \(\tau_P=e^{-\Omega(\Lambda)}\)
\State Set
\(\widetilde P\gets\Pi p_P(\cL)\Pi\)
as in \eqref{eq:Ptilde-definition}
\State Represent \(\widetilde S\) by replacing \(P\) with
\(\widetilde P\) in \eqref{eq:S-representation}
\State Construct the degree-\(d_E\) polynomial operators
\(\widetilde U_j\) and \(\widetilde V_j\) from \eqref{eq:Qpm} and
\eqref{eq:dE},
using \(\widetilde S\), the upward shift, and the projections onto
\(\cU\) described above
\For{\(v\in V(G)\)}
    \State Set column \(v\) of \(\widetilde{\overline Y}_j\) to
    \(\widetilde P
      (\rho_+^{-1}\widetilde U_j-\rho_-^{-1}\widetilde V_j)
      \widetilde P\1_v\)
\EndFor
\State Initialize \(\widetilde\sigma_i\gets0\) for every \(i\in[q]\)
\ForAll{\(\{u,v\}\in M_i\), over all \(i\in[q]\)}
    \State Add
    \(qw(uv)\bigl(
      (\widetilde{\overline Y}_j)_{uu}/\deg_G(u)
      +(\widetilde{\overline Y}_j)_{vv}/\deg_G(v)
      -2(\widetilde{\overline Y}_j)_{uv}/
       \sqrt{\deg_G(u)\deg_G(v)}
    \bigr)\)
    to \(\widetilde\sigma_i\)
\EndFor
\State \Return \(\widetilde\sigma_1,\ldots,\widetilde\sigma_q\)
\end{algorithmic}
\end{algorithm}

For the analysis, define the exact matrix
\[
    \overline Y_j
    :=
    P\left(\rho_+^{-1}U_j-\rho_-^{-1}V_j\right)P.
\]
\cref{alg:sparse-score-evaluation} explicitly materializes the
approximation \(\widetilde{\overline Y}_j\), one column at a time.  Each
column is obtained by applying \(\widetilde P\), the two polynomial
exponential approximations, and then \(\widetilde P\) again.  Forming all
\(n\) columns therefore costs
\begin{equation}
\label{eq:one-sparse-iteration}
    O\!\left(
        nm\phi^{-1}\Lambda
        \left[
            \sqrt{\eps^{-1}\log(n)\Lambda}
            +\Lambda
        \right]
    \right)
\end{equation}
time and \(O(n^2)\) space.

Once \(\widetilde{\overline Y}_j\) has been materialized, initialize the
\(q\) scores in \(O(q)\) time and scan every edge once.  Since the
matchings partition the edge set, the scan costs \(O(m+q)=O(m)\).  This is absorbed by
\eqref{eq:one-sparse-iteration}.  For the exact matrices, the identity
underlying these updates is
\begin{equation}
\label{eq:sparse-score-edge-scan}
\begin{aligned}
    \ip{W_j}{A_i}
    =q\sum_{\{u,v\}\in M_i}
    w(uv)\biggl(
        \frac{(\overline Y_j)_{uu}}{\deg_G(u)}
        +\frac{(\overline Y_j)_{vv}}{\deg_G(v)}
        -\frac{2(\overline Y_j)_{uv}}
                    {\sqrt{\deg_G(u)\deg_G(v)}}
    \biggr).
\end{aligned}
\end{equation}

It remains to verify the precision.  The construction gives
\[
    \norm{\widetilde P-P}=e^{-\Omega(\Lambda)}
\]
and
\[
    \norm{\widetilde U_j-U_j},
    \norm{\widetilde V_j-V_j}
    =
    e^{-C\Lambda}.
\]
Consequently, \eqref{eq:Y-assembly-error} yields
\[
    \norm{\widetilde{\overline Y}_j-\overline Y_j}
    =
    e^{-\Omega(\Lambda)}\phi^{-2}.
\]
Since
\(\Lambda\ge\log(256nk/(\eps\phi))\), the constants in the approximation
accuracies can be chosen so that this is at most \(\tau_Y\) from
\eqref{eq:tau-Y}.  The common reduction
\eqref{eq:score-error-reduction} then gives the accuracy required by
\cref{lem:inexact-scores}.

The per-iteration score scan costs \(O(m)\), and maintaining \(K\)
throughout the algorithm costs \(O(kn)\).  Both are dominated by the
matrix-materialization work.  Multiplying
\eqref{eq:one-sparse-iteration} by the
\(k=O(R\eps^{-2}\log(n))\) iterations therefore gives
\begin{equation}
\label{eq:sparse-expander-time}
O\!\left(
    nm\phi^{-1}R\eps^{-2}\log(n)\Lambda
    \left[
        \sqrt{\eps^{-1}\log(n)\Lambda}
        +\Lambda
    \right]
\right).
\end{equation}
time.

\subsection{Degree Regularization and Matching Decomposition}
The eigenvalue bound in \eqref{eq:R-general} is useful only when the number of
matchings is comparable to the minimum weighted degree.  A graph may fail
this condition because a high-degree vertex forces many matchings even
when other vertices have small degree.  We remove this obstruction by
splitting each vertex into a collection of vertices, called a fiber, of approximately equal-load copies.
An auxiliary expander inside each fiber prevents a low-conductance cut
from separating the copies, and self-loops equalize their weighted
degrees without changing the Laplacian.  Contracting the fibers after
sparsification recovers the original graph.

\begin{lemma}[Degree regularization]
\label{lem:degree-regularization}
Let $G$ be a weighted multigraph, possibly with self-loops, such that
every non-loop edge has weight in $[1,2]$, $\Phi(G)\ge\phi$, and
$\delta = \Omega(d)$ for some $d \ge 1$.  In
$O((m+\vol_G(V))\log n)$ time one can construct a weighted multigraph
$G^\circ$ and a partition of $V(G^\circ)$ into fibers
$(F_v)_{v\in V(G)}$ with the following properties:
\begin{enumerate}[label=(\roman*)]
\item $|V(G^\circ)|=O(n+\vol_G(V)/d)$ and
$|E(G^\circ)|=O(m+\vol_G(V))$;
\item every non-loop edge of \(G^\circ\) has weight in \([1,2]\), and
every \(x\in V(G^\circ)\) satisfies
\[
    cd\le\deg_{G^\circ}(x)\le Cd,
\]
for absolute constants $c,C>0$.  In particular,
$\deg_{G^\circ}^\#(x)\le\deg_{G^\circ}(x)=O(d)$;
\item $\Phi(G^\circ)=\Omega(\phi)$;
\item contracting every fiber maps a spectral sparsifier of $G^\circ$ to
a spectral sparsifier of $G$ with the same approximation factor.
\end{enumerate}
\end{lemma}

\begin{algorithm}[H]
\caption{\(\textsc{RegularizeDegrees}(G,d)\)}
\label{alg:degree-regularization}
\begin{algorithmic}[1]
\Require A graph \(G\) satisfying the hypotheses of
\cref{lem:degree-regularization}
\Ensure A graph \(G^\circ\) and fibers \((F_v)_{v\in V(G)}\)
\State \(s\gets\lceil d\rceil\)
\For{\(v\in V(G)\)}
    \State \(t_v\gets\max\{1,\lceil\deg_G(v)/s\rceil\}\)
    \State Create a fiber \(F_v\) of \(t_v\) vertices
    \If{\(t_v\ge2\)}
        \State Produce a constant-degree expander with constant expansion to be placed on \(F_v\)
    \EndIf
    \State Distribute the non-loop incidences at \(v\) among \(F_v\) so
    that every copy receives total weight \(O(s)\)
\EndFor
\State Initialize \(G^\circ\) on the vertex set
\(\mathbin{\dot\bigcup}_{v\in V(G)}F_v\)
\ForAll{non-loop edges \(e=uv\in E(G)\)}
    \State Add an edge of weight \(w(e)\) between the copies assigned the
    two incidences of \(e\)
\EndFor
\ForAll{fibers \(F_v\) with \(|F_v|\ge2\)}
    \State Add \(as\) parallel unit-weight copies of every edge of the
    fiber expander, for a sufficiently large absolute constant \(a\)
\EndFor
\For{\(x\in V(G^\circ)\)}
    \State Add a self-loop that raises \(\deg_{G^\circ}(x)\) to \(C_0s\),
    for a sufficiently large absolute constant \(C_0\)
\EndFor
\State \Return \(G^\circ,(F_v)_{v\in V(G)}\)
\end{algorithmic}
\end{algorithm}

\begin{proof}
Run \cref{alg:degree-regularization}.  Since \(d\ge1\), the parameter
\(s=\lceil d\rceil\) satisfies \(s=\Theta(d)\).

\paragraph{Size and degrees.}
The balanced incidence assignment exists by greedily placing the next
incidence on the least-loaded copy.  Since every incidence has weight at
most \(2\), the load of every copy is at most the average load plus \(2\),
and is therefore \(O(s)\).  Use an explicit constant-degree,
constant-conductance expander family for the fibers, constructible in
linear time~\cite[Theorem~2.4]{CKLPGS22}.

Every copy is incident to original and fiber edges of total weight
$O(s)$ before the self-loops are added, so the padding is nonnegative for
a suitable $C_0$.
Moreover, the lower bound on the input degrees gives
\[
    \sum_vt_v
    =O\!\left(n+\frac{\vol_G(V)}d\right)
    \quad\text{and}\quad
    dn=O(\vol_G(V)).
\]
It follows that the number of vertices and the total number of edge
copies satisfy (i): the auxiliary expanders contribute
\(O(s\sum_vt_v)=O(\vol_G(V))\) copies, and the padding contributes one
self-loop per new vertex.  The original non-loop edges retain their
weights, while the auxiliary non-loop edges have unit weight.  The
construction of the fibers and the degree padding therefore gives (ii).
We also have
\begin{equation}
\label{eq:fiber-volume}
    \vol_{G^\circ}(F_v)=\Theta(\deg_G(v))
    \qquad\text{for every }v\in V(G).
\end{equation}

\paragraph{Conductance.}
Fix
$S\subseteq V(G^\circ)$ with
$\vol_{G^\circ}(S)\le\vol_{G^\circ}(V(G^\circ))/2$.  For every $v$, let
$a_v:=\min\{|S\cap F_v|,|F_v\setminus S|\}$ and put
$A:=s\sum_va_v$.  The fiber expanders contribute boundary weight
$\Omega(A)$ to $E_{G^\circ}(S,V(G^\circ)\setminus S)$.

Round each fiber to its majority side, let $\widehat S$ be the resulting
union of fibers, and define
$T:=\{v:F_v\subseteq\widehat S\}$.  Since every vertex of $G^\circ$ has
degree $\Theta(s)$,
\[
    \vol_{G^\circ}(S\triangle\widehat S)=O(A).
\]
Every original edge of $E_G(T,V(G)\setminus T)$ that does not cross $S$
has an endpoint copy in $S\triangle\widehat S$.  Each copy is incident to
original edges of total weight $O(s)$, and hence
\begin{equation}
\label{eq:regularization-rounding-cut}
    w(E_{G^\circ}(S,V(G^\circ)\setminus S))
    \ge w(E_G(T,V(G)\setminus T))-O(A).
\end{equation}
By \eqref{eq:fiber-volume}, for every $X\subseteq V(G)$,
\begin{equation}
\label{eq:regularization-rounding-volume}
    \vol_{G^\circ}\!\left(\bigcup_{v\in X}F_v\right)
    =\Theta(\vol_G(X)).
\end{equation}

Write $b:=\vol_{G^\circ}(S)$.  If $A\ge c_1\phi b$ for a sufficiently
small absolute constant $c_1$, the fiber edges already give boundary
$\Omega(\phi b)$.  Otherwise
$\vol_{G^\circ}(S\triangle\widehat S)\le b/4$.  Since both $S$ and its
complement have volume at least $b$, both $\widehat S$ and its complement
have volume at least $3b/4$.  Equation
\eqref{eq:regularization-rounding-volume} therefore gives
\[
    \min\{\vol_G(T),\vol_G(V(G)\setminus T)\}=\Omega(b).
\]
Using $\Phi(G)\ge\phi$ in
\eqref{eq:regularization-rounding-cut}, and then taking $c_1$ sufficiently
small, yields
\[
    w(E_{G^\circ}(S,V(G^\circ)\setminus S))
    \ge\Omega(\phi b)-O(A)
    =\Omega(\phi b).
\]
This proves (iii).

\paragraph{Contraction.}
Let $Q:\R^{V(G)}\to\R^{V(G^\circ)}$ copy the coordinate of $v$
to every vertex of $F_v$.  Fiber edges and self-loops have zero energy on
$\operatorname{im}Q$, while the distributed original edges have exactly
their original energy.  Thus $Q^TL_{G^\circ}Q=L_G$.  Conjugating a
spectral-approximation inequality for $G^\circ$ by $Q$ proves (iv).
\end{proof}

\begin{lemma}[Matching decomposition]
\label{lem:matching-decomposition}
The non-loop edges of a multigraph can be partitioned deterministically into $O(\Delta^\#)$
matchings in $O(m\log(2\Delta^\#))$ time.
\end{lemma}

\begin{algorithm}[H]
\caption{\(\textsc{MatchingDecomposition}(G)\)}
\label{alg:matching-decomposition}
\begin{algorithmic}[1]
\Require A multigraph \(G\) of maximum non-loop combinatorial degree \(\Delta^\#\)
\Ensure A partition of \(E^\times(G)\) into \(O(\Delta^\#)\) matchings
\State \(\mathcal H\gets\{G\}\)
\While{some \(Q\in\mathcal H\) has maximum combinatorial degree greater than \(4\)}
    \State \(\mathcal H'\gets\emptyset\)
    \For{\(Q\in\mathcal H\)}
        \If{\(\Delta^\#(Q)\le4\)}
            \State Add \(Q\) to \(\mathcal H'\)
        \Else
            \ForAll{connected components \(C\) of \(Q\)}
                \State Add a dummy vertex adjacent to every odd-degree
                vertex of \(C\)
                \State Traverse an Euler tour and color its edges
                alternately red and blue
                \State Delete the dummy edges
            \EndFor
            \State Add the red and blue spanning subgraphs of \(Q\) to
            \(\mathcal H'\)
        \EndIf
    \EndFor
    \State \(\mathcal H\gets\mathcal H'\)
\EndWhile
\State \(\mathcal M\gets\emptyset\)
\For{\(Q\in\mathcal H\)}
    \State Greedily edge-color \(Q\) with at most seven colors
    \State Add each nonempty color class to \(\mathcal M\)
\EndFor
\State \Return \(\mathcal M\)
\end{algorithmic}
\end{algorithm}

\begin{proof}
Consider one split performed by
\cref{alg:matching-decomposition} on a graph of maximum combinatorial degree
$D$.  After the dummy edges are added, every degree is even.  Alternating
colors along an Euler tour makes the red and blue degrees at each original
vertex differ by at most two, and deleting the dummy edges preserves
\[
    \Delta^\#(H_{\mathrm{red}}),\Delta^\#(H_{\mathrm{blue}})
    \le D/2+2.
\]
If \(D_h\) is the largest maximum degree after \(h\) splitting rounds,
then
\[
    D_{h+1}\le\frac{D_h}{2}+2,
    \qquad
    D_h\le4+\frac{\Delta^\#-4}{2^h}.
\]
For \(h=\lceil\log_2(2\Delta^\#)\rceil\), the right-hand side is less
than \(5\); since \(D_h\) is an integer, every remaining graph has maximum
degree at most \(4\).  There are at most \(2^h=O(\Delta^\#)\) such graphs.
An edge in one of them is adjacent to at most six other edges, so greedy
edge coloring uses at most seven colors.  Each color class is a matching,
giving $O(\Delta^\#)$ matchings in total.  Every original edge is processed
$O(1)$ times per splitting round, which proves the
$O(m\log(2\Delta^\#))$ running time.
\end{proof}

\begin{algorithm}[H]
\caption{\(\textsc{ExpanderSparsify}(G,d,\phi,\eps,\mathtt{mode})\)}
\label{alg:expander-sparsify}
\begin{algorithmic}[1]
\Require A graph \(G\) whose non-loop weights lie in \([1,2]\), with
\(d\ge1\), \(\Phi(G)\ge\phi\), and \(\delta = \Omega(d)\), and
\(\mathtt{mode}\in\{\mathtt{dense},\mathtt{sparse}\}\)
\Ensure A \((1\pm\eps)\)-spectral sparsifier of \(G\)
\State \((G^\circ,(F_v)_v)\gets\Call{RegularizeDegrees}{G,d}\)
\State \((M_1,\ldots,M_q)\gets
\Call{MatchingDecomposition}{G^\circ}\)
\State Let \(\phi^\circ=c\phi\), where \(c>0\) is the absolute constant
from \cref{lem:degree-regularization}
\State \(H^\circ\gets
\Call{MatchingSparsifier}
{G^\circ,M_1,\ldots,M_q,2,\phi^\circ,\eps}\), evaluating scores by
\(\mathtt{mode}\)
\State Contract every fiber \(F_v\) in \(H^\circ\) and delete self-loops
\State \Return the resulting graph \(H\)
\end{algorithmic}
\end{algorithm}

\begin{lemma}[Expander sparsification]
\label{lem:expander-sparsification}
Let $G$ be a weighted multigraph whose non-loop
edge weights lie in $[1,2]$.  Suppose $\Phi(G)\ge\phi$ and
$\delta = \Omega(d)$ for some $d\ge1$.  Put
$n_0:=n+\vol_G(V)/d$ and $m_0:=m+\vol_G(V)$.
For $0<\eps\le1/2$, define \(L\) by \eqref{eq:L-definition} using
ambient upper bounds on the vertex and edge counts.  Then one can
construct a $(1\pm\eps)$-spectral sparsifier
with
\begin{equation}
\label{eq:expander-cor-size}
    O\!\left(n_0\phi^{-2}\eps^{-2}L\right)
\end{equation}
edges.  The dense implementation takes
\begin{equation}
\label{eq:expander-cor-dense}
    O\!\left(
        m_0 +
        (n_0^\omega\log n_0+n_0^2L^{O(1)}+m_0)
        \phi^{-2}\eps^{-2}L
    \right)
\end{equation}
time, while the sparse implementation takes
\begin{equation}
\label{eq:expander-cor-sparse}
    O\!\left(
        m_0+n_0m_0\phi^{-3}\eps^{-5/2}L^3
    \right).
\end{equation}
time.
\end{lemma}

\begin{proof}
Run \cref{alg:expander-sparsify}, and write
\[
    N:=|V(G^\circ)|,
    \qquad
    M:=|E(G^\circ)|.
\]
By \cref{lem:degree-regularization},
\[
    N=O(n_0),
    \qquad
    M=O(m_0),
\]
every weighted degree in \(G^\circ\) is \(\Theta(d)\), its maximum
combinatorial degree is \(O(d)\), and
\(\Phi(G^\circ)=\Omega(\phi)\).  Its non-loop edge weights also remain in
\([1,2]\).

\paragraph{Number of selected matchings and correctness.}
By \cref{lem:matching-decomposition}, the non-loop edges of \(G^\circ\)
can be partitioned into \(q=O(d)\) matchings.  In
\eqref{eq:R-general}, we therefore have
\[
    \beta\le2,
    \qquad
    \delta =\Theta(d),
    \qquad
    R=O(\phi^{-2}).
\]
Consequently,
\[
    k=O\!\left(\phi^{-2}\eps^{-2}L\right)
\]
matchings are selected.  The analysis of
\cref{alg:matching-sparsifier} shows that the reweighted union
\(H^\circ\) is a \((1\pm\eps)\)-spectral sparsifier of \(G^\circ\).
Property (iv) of \cref{lem:degree-regularization} then shows that
contracting every fiber in \(H^\circ\) and deleting self-loops produces a
\((1\pm\eps)\)-spectral sparsifier \(H\) of \(G\).

\paragraph{Output size.}
Every selected matching contains at most \(N/2\) edges, so
\eqref{eq:expander-size} and \(R=O(\phi^{-2})\) give
\[
    |E(H)|
    \le |E(H^\circ)|
    =
    O\!\left(
        N\phi^{-2}\eps^{-2}\log N
    \right)
    =
    O\!\left(
        n_0\phi^{-2}\eps^{-2}L
    \right).
\]
This proves \eqref{eq:expander-cor-size}.

\paragraph{Dense implementation.}
One dense score evaluation takes
\[
    O\!\left(
        M+N^\omega\log N+N^2L^{O(1)}
    \right)
    =
    O\!\left(
        m_0+n_0^\omega\log n_0+n_0^2L^{O(1)}
    \right)
\]
time.  Multiplying by
\(k=O(\phi^{-2}\eps^{-2}L)\) and including degree
regularization and the matching decomposition proves
\eqref{eq:expander-cor-dense}.

\paragraph{Sparse implementation.}
The definition of \(L\), together with
\(k=O(\phi^{-2}\eps^{-2}L)\), gives \(\Lambda=O(L)\).
Substituting
\[
    N=O(n_0),
    \qquad
    M=O(m_0),
    \qquad
    R=O(\phi^{-2}),
    \qquad
    \Lambda=O(L)
\]
into \eqref{eq:sparse-expander-time}, the expression in square brackets
is \(O(\eps^{-1/2}L)\).  Since
\(\phi^{-1}R=O(\phi^{-3})\), the total sparse score-evaluation time is
\[
    O\!\left(
        n_0m_0\phi^{-3}\eps^{-5/2}L^3
    \right).
\]
The regularization, matching decomposition, and contraction costs are
absorbed by this bound, proving \eqref{eq:expander-cor-sparse}.
\end{proof}

\section{Global Sparsification}
\label{sec:global-sparsification}
The preceding routine applies to one approximately regular expander.  To
handle an arbitrary graph with weights in $[1,2]$, we repeatedly expose
edge-disjoint subgraphs satisfying those hypotheses.  Low-degree peeling
first guarantees the minimum degree needed by regularization while
discarding only a constant fraction of the current edges.  A deterministic
expander decomposition then places a constant fraction of the surviving
edges inside high-conductance clusters.  We sparsify the cluster graphs
and continue on the remaining edges.  The residual edge count decreases
geometrically, and edge-disjoint composition turns all local guarantees
into one global spectral approximation.

For a cluster \(U\subseteq V(G)\), let \(G\langle U\rangle\) be obtained
from \(G[U]\) by adding
a self-loop of weight \(\deg_G(v)-\deg_{G[U]}(v)\) at every \(v\in U\).  These loops
retain the degree in the whole graph without changing the Laplacian of the cluster.  Thus
\(\deg_{G\langle U\rangle}(v)=\deg_G(v)\).

The following deterministic expander decomposition algorithm follows from
\cite[Theorem~2.10]{Chuzhoy23} by adding $\deg(v)$ self-loops for every $v \in V$
before applying the decomposition and then adjusting constants.

\begin{theorem}[Deterministic expander decomposition]
\label{thm:expander-decomposition}
There is a deterministic $m^{1+o(1)}$-time algorithm that, given an
$n$-vertex, $m$-edge unweighted graph $G$, computes a partition
\[
    V(G)=V_1\mathbin{\dot\cup}\cdots\mathbin{\dot\cup}V_t
\]
such that at most $m/2$ edges have endpoints in distinct clusters and,
for every $i$, the graph
$G\langle V_i\rangle$ has conductance at least
$\phi_n:=\log^{-9-o(1)}n$.
\end{theorem}

\begin{algorithm}[H]
\caption{\(\textsc{GlobalSparsify}(G,\eps,\mathtt{mode})\)}
\label{alg:global-sparsify}
\begin{algorithmic}[1]
\Require A multigraph \(G=(V,E,w)\) whose non-loop weights lie in
\([1,2]\), an accuracy \(0<\eps\le1/2\), and
\(\mathtt{mode}\in\{\mathtt{dense},\mathtt{sparse}\}\)
\Ensure A \((1\pm\eps)\)-spectral sparsifier of \(G\)
\State \(R\gets E^\times(G)\) and \(H\gets\emptyset\)
\While{\(|R|\ge n\)}
    \State \(F\gets(V,R)\), \(m_F\gets|R|\), and
    \(d\gets2w(R)/n\)
    \State \(U\gets V\)
    \While{some \(v\in U\) satisfies
    \(\deg_{F[U]}(v)<d/10\)}
        \State \(U\gets U\setminus\{v\}\)
    \EndWhile
    \State Compute \(V_1\mathbin{\dot\cup}\cdots
    \mathbin{\dot\cup}V_t=U\) by applying
    \cref{thm:expander-decomposition} to \(F[U]\) without weights
    \For{\(i=1,\ldots,t\)}
        \If{\(E^\times(F[V_i])\ne\emptyset\)}
            \State \(Q_i\gets F[U]\langle V_i\rangle\)
            \State \(H_i\gets
            \Call{ExpanderSparsify}
            {Q_i,d,\phi_n/2,\eps,\mathtt{mode}}\)
            \State \(H\gets H\mathbin{\dot\cup}H_i\)
        \EndIf
    \EndFor
    \State \(R\gets R\setminus
    \biguplus_{i=1}^tE^\times(F[V_i])\)
\EndWhile
\State Add every edge of \(R\), with its original weight, to \(H\)
\State \Return \(H\)
\end{algorithmic}
\end{algorithm}

\begin{theorem}[Deterministic sparsification at one weight scale]
\label{thm:global-sparsification}
Let $G$ be an $n$-vertex, $m$-edge weighted multigraph whose non-loop
edge weights lie in $[1,2]$, and let $0<\eps\le1/2$.  Suppose
$m,\eps^{-1}\le n^{O(1)}$.  Then there is a deterministic algorithm that constructs a
$(1\pm\eps)$-spectral sparsifier with
\begin{equation}
\label{eq:global-sparsifier-size}
    O\!\left(n\phi_n^{-2}\eps^{-2} \log^2 n\right)
\end{equation}
edges.  Its dense implementation runs in
\begin{equation}
\label{eq:dense-global-time}
    m^{1+o(1)}
    +O\!\left(
        (n^\omega \log^3 n +m \log n) \phi_n^{-2}\eps^{-2}
    \right)
\end{equation}
time, while its sparse implementation runs in
\begin{equation}
\label{eq:sparse-global-time}
    m^{1+o(1)}
    +O\!\left(
        nm\phi_n^{-3}\eps^{-5/2}\log^3 n
    \right)
\end{equation}
time.

Since $\phi_n^{-1}=\log^{9+o(1)}n$, these bounds are
$O(n\eps^{-2}\log^{20+o(1)}(n/\eps))$ edges,
\[
    m^{1+o(1)}
    +O\!\left(
        n^\omega\eps^{-2}\log^{21+o(1)}(n/\eps)
        +m\eps^{-2}\log^{19+o(1)}(n/\eps)
    \right)
\]
time using the dense implementation, and
\[
    m^{1+o(1)}
    +O\!\left(
        nm\eps^{-5/2}\log^{30+o(1)}(n/\eps)
    \right)
\]
time using the sparse implementation.
\end{theorem}

\begin{proof}
Run \cref{alg:global-sparsify}.  Consider an iteration with residual
graph \(F=(V,R)\), and write
\[
    m_F:=|R|,
    \qquad
    W_F:=w(R).
\]
Since every residual edge has weight in \([1,2]\),
\[
    m_F\le W_F\le2m_F.
\]
The loop condition gives \(m_F\ge n\), so
\[
    d:=\frac{2W_F}{n}\ge2.
\]

We first analyze the low-degree peeling step.  Charge an edge, with its
weight, when its first endpoint is deleted.  When a vertex is deleted,
its current weighted degree is less than \(d/10\); hence the total charged
weight is less than
\[
    \frac{nd}{10}
    =
    \frac{W_F}{5}.
\]
Every deleted edge is charged exactly once.  Since every edge has weight
at least \(1\), the number of deleted edges is less than
\[
    \frac{W_F}{5}
    \le
    \frac{2m_F}{5}.
\]
It follows that
\[
    |E(F[U])|
    \ge
    m_F-\frac{W_F}{5}
    \ge
    \frac{3m_F}{5}.
\]
Moreover, upon termination of the peeling step,
\[
    \deg_{F[U]}(v)\ge\frac d{10}
    \qquad\text{for every }v\in U.
\]

Let \(F^\#\) denote the unweighted multigraph obtained from \(F\) by
assigning unit weight to every edge while retaining parallel copies.
The expander decomposition applied to \(F^\#[U]\) leaves at least half of
the edges of \(F^\#[U]\) inside its clusters.  Therefore
\[
    \sum_{i=1}^t |E(F[V_i])|
    \ge
    \frac12|E(F[U])|
    \ge
    \frac{3m_F}{10}.
\]
Thus every iteration removes and sparsifies at least \(3m_F/10\)
residual edges.

For each cluster \(V_i\), let
\[
    Q_i:=F[U]\langle V_i\rangle.
\]
The self-loops preserve the degree, so
\[
    \deg_{Q_i}(v)=\deg_{F[U]}(v)
    \qquad\text{for every }v\in V_i.
\]
Consequently,
\[
    \delta_{Q_i}\ge\frac d{10}=\Omega(d).
\]

The expander decomposition guarantees
\[
    \Phi\!\left(F^\#[U]\langle V_i\rangle\right)\ge\phi_n.
\]
For every \(S\subseteq V_i\), the weights in \([1,2]\) give
\[
    w(E_{Q_i}(S,V_i\setminus S))
    \ge
    |E_{F^\#}(S,V_i\setminus S)|
\]
and
\[
    \vol_{Q_i}(S)
    =
    \sum_{v\in S}\deg_{F[U]}(v)
    \le
    2\sum_{v\in S}\deg^\#_{F[U]}(v)
    =
    2\vol_{F^\#[U]\langle V_i\rangle}(S).
\]
The same inequality holds for \(V_i\setminus S\).  It follows that
\[
    \Phi(Q_i)
    \ge
    \frac12
    \Phi\!\left(F^\#[U]\langle V_i\rangle\right)
    \ge
    \frac{\phi_n}{2}.
\]
Thus every call to \(\textsc{ExpanderSparsify}\) satisfies its minimum
degree and conductance hypotheses.

Define
\[
    n_i
    :=
    |V_i|+\frac{\vol_{Q_i}(V_i)}d,
    \qquad
    m_i
    :=
    |E(F[V_i])|+\vol_{Q_i}(V_i).
\]
Since the clusters partition \(U\) and degrees are preserved,
\[
\begin{aligned}
    \sum_{i=1}^t \vol_{Q_i}(V_i)
    =
    \vol_{F[U]}(U)
    =
    2w(E(F[U]))
    \le
    2W_F
    =
    nd.
\end{aligned}
\]
Therefore
\[
    \sum_{i=1}^t n_i
    \le
    |U|+\frac{nd}{d}
    \le
    2n
    =
    O(n).
\]
Similarly, using \(W_F\le2m_F\),
\[
\begin{aligned}
    \sum_{i=1}^t m_i
    =
    \sum_{i=1}^t|E(F[V_i])|
      +\sum_{i=1}^t\vol_{Q_i}(V_i)
    \le
    m_F+2W_F
    \le
    5m_F
    =
    O(m_F).
\end{aligned}
\]

By \cref{lem:expander-sparsification}, the number of edges produced in
this iteration is
\[
\begin{aligned}
    \sum_{i=1}^t
    O\!\left(
        n_i\phi_n^{-2}\eps^{-2}L
    \right)
    &=
    O\!\left(
        n\phi_n^{-2}\eps^{-2}L
    \right).
\end{aligned}
\]

For the dense implementation, the total work of the expander
sparsification calls in one iteration is
\[
\begin{aligned}
O\!\left(
    \sum_{i=1}^t
    \left[
        m_i+
        \bigl(n_i^\omega L+n_i^2L^{O(1)}+m_i\bigr)
        \phi_n^{-2}\eps^{-2}L
    \right]
\right) = O\!\left(
\bigl(n^\omega L+m_F\bigr)
        \phi_n^{-2}\eps^{-2}L
    \right).
\end{aligned}
\]
Here
\(\sum_i n_i^2\le(\sum_i n_i)^2=O(n^2)\), and the resulting
\(n^2L^{O(1)}\) term is absorbed by \(n^\omega L\).

For the sparse implementation,
\[
    \sum_{i=1}^t n_i m_i
    \le
    \left(\sum_{i=1}^t n_i\right)
    \left(\sum_{i=1}^t m_i\right)
    =
    O(nm_F).
\]
Hence the work of all expander sparsification calls in one iteration is
\[
    O\!\left(
        nm_F\phi_n^{-3}\eps^{-5/2}L^3
    \right).
\]

The peeling, cluster construction, and other bookkeeping take
\(O(m_F)\) time in an iteration.  The deterministic expander
decomposition takes \(m_F^{1+o(1)}\) time.  Since every iteration removes
at least \(3/10\) of the residual edges, the residual edge counts decrease
geometrically.  Consequently, there are
\[
    O(\log(2m))=O(L)
\]
iterations and
\[
    \sum_F m_F=O(m).
\]
The total cost of all expander decompositions and bookkeeping is therefore
\(m^{1+o(1)}\).

Each iteration outputs
\(O(n\phi_n^{-2}\eps^{-2}L)\) edges, and there are \(O(L)\) iterations.
The final residual graph contains fewer than \(n\) edges and is included
without modification.  Thus the total number of output edges is
\[
    O\!\left(
        n\phi_n^{-2}\eps^{-2}L^2
    \right),
\]
which gives \eqref{eq:global-sparsifier-size} because \(L=O(\log n)\).

For the dense implementation, the \(n^\omega\)-dependent term is incurred
once per iteration.  The total running time is therefore
\[
\begin{aligned}
    m^{1+o(1)}
    &+
    O\!\left(
        n^\omega\phi_n^{-2}\eps^{-2}L^3
        +
        m\phi_n^{-2}\eps^{-2}L
    \right),
\end{aligned}
\]
which is \eqref{eq:dense-global-time}.  For the sparse implementation,
summing over the geometrically decreasing residual edge counts gives
\[
    m^{1+o(1)}
    +
    O\!\left(
        nm\phi_n^{-3}\eps^{-5/2}L^3
    \right),
\]
which is \eqref{eq:sparse-global-time}.

Finally, the graphs \(F[V_i]\) sparsified in the different iterations,
together with the final residual graph, partition the original edge
multiset. Thus,
\[
    L_{Q_i}=L_{F[V_i]}.
\]
Each \(H_i\) is therefore a \((1\pm\eps)\)-spectral sparsifier of the
corresponding edge-induced piece \(F[V_i]\).  Applying
\cref{lem:composition} to these edge-disjoint pieces and the final
residual graph proves that the returned graph \(H\) is a
\((1\pm\eps)\)-spectral sparsifier of \(G\).
\end{proof}

\subsection{Weighted Graphs}
\label{sec:weighted}
The one-scale algorithm requires all non-loop weights to differ by at most
a factor of two.  We remove this restriction by decomposing a general
weighted graph into dyadic weight classes.  These classes are
edge-disjoint so sparsifying
every class with error \(\eps\) and taking the union still gives error
\(\eps\) by \cref{lem:composition}.  Only the vertex-dependent work and output size are repeated
across classes; all edge-dependent terms sum to their value on the
original graph.

For a graph $G$ with positive weights, let
\[
    b(G):=
    1+\left\lceil
        \log_2\frac{w_{\max}}{w_{\min}}
    \right\rceil.
\]
Thus the non-loop edges of $G$ occupy at most $b(G)$ dyadic weight
classes.

\begin{algorithm}[H]
\caption{\(\textsc{WeightedSparsify}(G,\eps,\mathtt{mode})\)}
\label{alg:weighted-sparsify}
\begin{algorithmic}[1]
\Require A positively weighted graph \(G\), an accuracy
\(0<\eps\le1/2\), and
\(\mathtt{mode}\in\{\mathtt{dense},\mathtt{sparse}\}\)
\Ensure A \((1\pm\eps)\)-spectral sparsifier of \(G\)
\State Partition the non-loop edges into
\(E_a:=\{e:2^a\le w(e)<2^{a+1}\}\)
\State \(H\gets\emptyset\)
\ForAll{nonempty classes \(E_a\)}
    \State Let \(G_a\) contain \(E_a\), with every edge weight divided by
    \(2^a\)
    \State \(H_a\gets
    \Call{GlobalSparsify}{G_a,\eps,\mathtt{mode}}\)
    \State Multiply every edge weight of \(H_a\) by \(2^a\)
    \State \(H\gets H\mathbin{\dot\cup}H_a\)
\EndFor
\State \Return \(H\)
\end{algorithmic}
\end{algorithm}

\begin{theorem}[Deterministic sparsification of weighted graphs]
\label{thm:weighted-sparsification}
Let $G$ be an $n$-vertex, $m$-edge positively weighted multigraph with
$b:=b(G)$, and let $0<\eps\le1/2$.  Suppose
$m,\eps^{-1}\le n^{O(1)}$.  There is a deterministic algorithm
that constructs a $(1\pm\eps)$-spectral sparsifier with
\begin{equation}
\label{eq:weighted-sparsifier-size}
    O\!\left(
        n b \phi_n^{-2}\eps^{-2} \log^2 n
    \right)
\end{equation}
edges.  Its dense implementation runs in
\begin{equation}
\label{eq:weighted-dense-time}
    m^{1+o(1)}
    +O\!\left(
        (n^\omega b \log^3 n +m \log n) \phi_n^{-2}\eps^{-2}
    \right)
\end{equation}
time, while its sparse implementation runs in
\begin{equation}
\label{eq:weighted-sparse-time}
    m^{1+o(1)}
    +O\!\left(
        nm\phi_n^{-3}\eps^{-5/2} \log^3 n
    \right)
\end{equation}
time.

Parallel edges may be aggregated before running the algorithm.
If the resulting graph has $\bar m$ edges and $\bar b$ dyadic weight
classes, the same conclusions hold with $m,b$ replaced by
$\bar m,\bar b$, together with an initial $O(m)$ aggregation cost.  In
particular, $\bar m\le\min\{m,\binom n2\}$ after self-loops are deleted,
and
\[
    \bar b\le b+\lceil\log_2m\rceil+1.
\]
\end{theorem}

\begin{proof}
Run \cref{alg:weighted-sparsify}.  Within each scaled class, every edge
weight lies in $[1,2)$, so \cref{thm:global-sparsification} applies.
Rescaling a class back by $2^a$ preserves its relative spectral error.
Since the classes partition the edge multiset, their union is a
$(1\pm\eps)$-spectral approximation of $G$ by
\cref{lem:composition}.

Each class contributes at most
$O(n\phi_n^{-2}\eps^{-2}L^2)$ output edges and one copy of the
vertex-dependent dense work, which proves
\eqref{eq:weighted-sparsifier-size} and the \(n^\omega b L^3\) term in
\eqref{eq:weighted-dense-time}.  If $m_a:=|E_a|$, then
$\sum_am_a=m$.  Consequently, the edge-dependent dense work and all the
sparse work sum over the classes without an additional factor $b$.
Moreover, the preprocessing costs
$\sum_am_a^{1+o(1)}\le m^{1+o(1)}$.  This proves
\eqref{eq:weighted-dense-time} and \eqref{eq:weighted-sparse-time}.

Aggregating parallel edges preserves the Laplacian.  It leaves at
most one non-loop edge for each unordered pair of vertices.  An aggregated
weight is the sum of at most $m$ original weights, so its base $2$ logarithm can
exceed the original range by at most $\lceil\log_2m\rceil+1$.  Applying
the preceding analysis to the aggregated graph proves the final claim.
\end{proof}

\begin{lemma}[Weight spread]
\label{lem:weight-spread}
Consider one invocation of the dense or sparse weighted sparsification
algorithm with accuracy $\eps$, size bounds $n,m$, and
conductance lower bound $\phi_n$.  Suppose
$m,\eps^{-1}\le n^{O(1)}$, and suppose the ratio between the largest and
smallest input edge weights is less than $2^b$.  After aggregating
parallel edges, the ratio between the largest and smallest output edge
weights is at most $2^{b+O(\log n)}$.
\end{lemma}

\begin{proof}
Fix one dyadic input class and scale it so that its edge weights lie in
$[1,2)$.  Every edge copy retained by
\cref{alg:matching-sparsifier} is multiplied by $q/k$.  In every
nontrivial call,
\[
    1\le q\le n^{O(1)}
    \qquad\text{and}\qquad
    1\le k
    \le
    O\!\left(\phi_n^{-2}\eps^{-2}\log n\right)
    =
    n^{O(1)}.
\]
Hence both $q/k$ and $k/q$ are bounded by $n^{O(1)}$ over all
clusters and peeling rounds.  Before aggregation, every output copy
originating in this class therefore has weight in
$[n^{-O(1)},n^{O(1)}]$ in the scaled units.

Aggregating parallel output copies can increase an edge weight by at most
the total number of output copies, which is also $n^{O(1)}$ under the
stated size and accuracy assumptions.  Undoing the dyadic scaling shows
that every final output weight lies between
$n^{-O(1)}w_{\min}$ and $n^{O(1)}w_{\max}$.  Its dyadic span is therefore
\[
    1+\left\lceil
        \log_2\frac{n^{O(1)}w_{\max}}
                       {n^{-O(1)}w_{\min}}
    \right\rceil
    \le b+O(\log n),
\]
as claimed.
\end{proof}

\section{Recursive Self-Improvement}
\label{sec:self-improvement}
This section formalizes the self-reduction described in the introduction.
At each level, the recursive calls sparsify an edge-disjoint decomposition
into small block-pair graphs.  Their union is already much sparser than
the input, so one application of the sparse implementation decreases the number of edges further and has low runtime.
We track three quantities through the recursion: the vertex
exponent, the additive growth of the dyadic weight span, and the
multiplicative accumulation of spectral error.  The exponent recurrence
has fixed point \(2\), while the other two quantities incur only
polylogarithmic losses over \(O(\log n)\) levels.

Define
\begin{equation}
\label{eq:alpha-recurrence}
    \alpha_0:=\omega,
    \qquad
    \alpha_{r+1}:=3-\frac{1}{\alpha_r-1}.
\end{equation}

\begin{algorithm}[H]
\caption{\(\mathcal A_r(G,\eta)\)}
\label{alg:recursive-sparsification}
\begin{algorithmic}[1]
\Require A positively weighted graph \(G\), a level \(r\ge0\), and an
accuracy \(0<\eta\le1/2\)
\Ensure A spectral approximation of \(G\)
\State Aggregate parallel edges and delete self-loops
\If{\(|V(G)|<\eta^{-1}\)}
    \State \Return \(G\)
\EndIf
\If{\(r=0\)}
    \State \Return
    \(\Call{WeightedSparsify}{G,\eta,\mathtt{dense}}\)
\EndIf
\State \(s\gets
\left\lceil n^{1/(\alpha_{r-1}-1)}\right\rceil\)
\State Partition \(V(G)\) into \(V_1,\ldots,V_t\), each of size at most
\(s\)
\State \(J\gets\emptyset\)
\For{\(1\le a\le c\le t\)}
    \State Let \(G_{ac}\) contain \(E^\times(G[V_a])\) if \(a=c\), and
    \(E_G(V_a,V_c)\) if \(a<c\)
    \If{\(E(G_{ac})\ne\emptyset\)}
        \State Delete isolated vertices from \(G_{ac}\)
        \State \(J_{ac}\gets\mathcal A_{r-1}(G_{ac},\eta)\)
        \State \(J\gets J\mathbin{\dot\cup}J_{ac}\)
    \EndIf
\EndFor
\State Aggregate parallel edges in \(J\)
\State \Return
\(\Call{WeightedSparsify}{J,\eta,\mathtt{sparse}}\)
\end{algorithmic}
\end{algorithm}

\begin{theorem}[The \(r\)-level algorithm]
\label{thm:r-level}
Let \(G\) be a positively weighted \(n\)-vertex, \(m\)-edge graph whose
non-loop edge weights occupy at most \(b\) dyadic classes.  Let
\(0<\eta\le1/2\) and \(r\ge0\), and suppose
\(m \le n^{O(1)}\).  Then \(\mathcal A_r(G,\eta)\) returns a graph
\(H\) satisfying
\[
    (1-\eta)^{r+1}L_G
    \preceq L_H
    \preceq(1+\eta)^{r+1}L_G.
\]
It has
\begin{equation}
\label{eq:r-level-size}
    O\!\left(
        n (b+(r+1)\log(n))
        \phi_n^{-2}\eta^{-2}\log^2(n)
    \right)
\end{equation}
edges, and the running time is
\begin{equation}
\label{eq:r-level-time}
\begin{aligned}
    (r+1)m^{1+o(1)}
    +
    O\!\left(
        n^{\alpha_r} (b+(r+1)\log(n))
        \phi_n^{-5}\eta^{-9/2}\log^5(n)
    \right).
\end{aligned}
\end{equation}
\end{theorem}

\begin{proof}
If the small-instance branch is taken, the algorithm returns an exact
representation of the input.  After aggregation it has at most \(n^2\)
edges, which is bounded by the claimed output-size expression when
\(n<\eta^{-1}\), and its processing time is also dominated by the stated
bound.  We may therefore assume that this branch is not taken.

We proceed by induction on \(r\).

For \(r=0\), the claim follows from
\cref{thm:weighted-sparsification} with its dense implementation after aggregating parallel edges so that $m \le n^2 \le n^\omega$. Aggregation may add \(O(\log n)\) dyadic
classes, which is covered by \(b+\log n\).

Suppose now that \(r\ge1\).  The block-pair graphs are edge-disjoint and
have at most \(2s\) vertices.  By induction and
\cref{lem:composition}, the union \(J\) of their sparsifiers satisfies
\[
    (1-\eta)^rL_G
    \preceq L_J
    \preceq(1+\eta)^rL_G.
\]
Let \(\mathcal B\) be the set of nonempty block-pair instances, and write
\(n_{ac}:=|V(G_{ac})|\) and \(m_{ac}:=|E(G_{ac})|\).  There are
\(O(n^2/s^2)\) such instances and \(n_{ac}\le2s\).  Using the ambient
quantities \(\log n\) and \(\phi_n^{-1}\) to upper-bound the corresponding
parameters of every recursive instance, the sum of their
vertex-dependent costs is
\begin{equation}
\label{eq:local-recursive-cost}
    O\!\left(
        n^2s^{\alpha_{r-1}-2}
        (b+r\log n)
        \phi_n^{-5}\eta^{-9/2}\log^5n
    \right).
\end{equation}

We now consider the edge-dependent term.  By induction, the call
on \(G_{ac}\) contributes \(r\,m_{ac}^{1+o(1)}\).  The block-pair graphs
are edge-disjoint, so \(\sum_{(a,c)\in\mathcal B}m_{ac}=m\).  Taking the
subpolynomial factor uniformly over inputs of size at most \(m\) gives
\[
    \sum_{(a,c)\in\mathcal B}
        r\,m_{ac}^{1+o(1)}
    \le
    r\,m^{o(1)}
        \sum_{(a,c)\in\mathcal B}m_{ac}
    =
    r\,m^{1+o(1)}.
\]

Applying \eqref{eq:r-level-size} to the recursive calls shows that
\begin{equation}
\label{eq:intermediate-size}
    |E(J)|
    =
    O\!\left(
        \frac{n^2}{s} (b+r\log(n))
        \phi_n^{-2}\eta^{-2}\log^2(n)
    \right).
\end{equation}
Let \(M:=|E(J)|\) after aggregation.  In addition to
\eqref{eq:intermediate-size}, we have \(M\le m\), because every recursive
output is supported on its corresponding block-pair graph and these
graphs are edge-disjoint.  The cost of the sparse
weighted call on \(J\) is therefore
\begin{equation}
\label{eq:recursive-compression-cost}
    M^{1+o(1)} + O\!\left(
        \frac{n^3}{s}
        (b+r\log n)
        \phi_n^{-5}\eta^{-9/2}\log^5n
    \right).
\end{equation}
We have that 
\(M^{1+o(1)}\le m^{1+o(1)}\).  Together with the
\(r\,m^{1+o(1)}\) contribution of the recursive calls, this gives the
\((r+1)m^{1+o(1)}\) term in \eqref{eq:r-level-time}.  The final sparse
call contributes one additional \((1\pm\eta)\) approximation, proving
the claimed spectral inequality.

Each sparsification and aggregation step enlarges the dyadic weight range
by at most \(O(\log n)\).  Hence the output occupies at most
\(b+O((r+1)\log n)\) dyadic classes, and
\cref{thm:weighted-sparsification} gives \eqref{eq:r-level-size}.

It remains to choose the block size.  The two vertex-dependent terms
\eqref{eq:local-recursive-cost} and
\eqref{eq:recursive-compression-cost} have the same lower order factors, so they are
balanced by
\(s=n^{1/(\alpha_{r-1}-1)}\).

We have
\[
    n^2s^{\alpha_{r-1}-2}
    =
    \Theta\!\left(
        n^{2+(\alpha_{r-1}-2)/(\alpha_{r-1}-1)}
    \right)
    =
    \Theta(n^{\alpha_r})
\]
and
\[
    \frac{n^3}{s}
    =
    \Theta\!\left(
        n^{3-1/(\alpha_{r-1}-1)}
    \right)
    =
    \Theta(n^{\alpha_r}).
\]

This proves
\eqref{eq:r-level-time}.
\end{proof}

\begin{lemma}[Convergence of the exponent]
\label{lem:alpha-closed-form}
For every \(r\ge0\),
\begin{equation}
\label{eq:alpha-closed-form}
    \alpha_r
    =
    2+\frac{1}{r+1/(\omega-2)}.
\end{equation}
In particular, \((\alpha_r)_{r\ge0}\) decreases to \(2\), and for every
\(r\ge\ln n\),
\[
    n^{\alpha_r}\le e\,n^2.
\]
\end{lemma}

\begin{proof}
Let \(\delta_r:=\alpha_r-2\), the gap between the current exponent and
the fixed point of $2$.  From \eqref{eq:alpha-recurrence},
\[
    \delta_{r+1}
    =
    1-\frac1{1+\delta_r}
    =
    \frac{\delta_r}{1+\delta_r}.
\]
Thus \(\delta_{r+1}\in(0,\delta_r)\), and taking reciprocals linearizes
the recurrence:
\[
    \frac1{\delta_{r+1}}
    =
    \frac1{\delta_r}+1.
\]
Since \(\delta_0=\omega-2\), iteration gives
\[
    \frac1{\delta_r}
    =
    r+\frac1{\omega-2},
\]
which is equivalent to \eqref{eq:alpha-closed-form} and also shows that
\(\alpha_r\rightarrow 2\) from above.

Finally, if \(r\ge\ln n\), then
\[
    n^{\alpha_r}
    =
    n^2
    \exp\!\left(
        \frac{\ln n}{r+1/(\omega-2)}
    \right)
    \le e\,n^2.
\]
\end{proof}

\Cref{fig:alpha-convergence} illustrates the \(1/r\) decay of the
exponent gap.

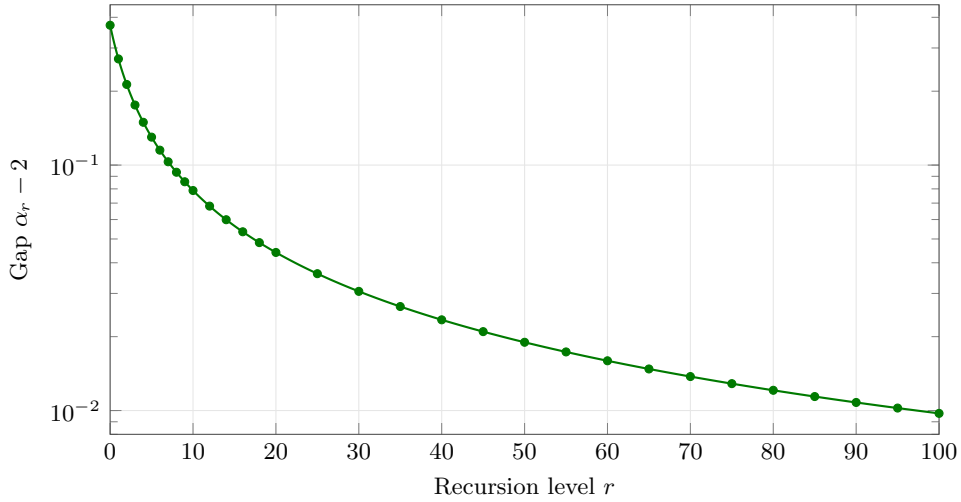
\begin{figure}[H]
\centering
\begin{tikzpicture}
\begin{semilogyaxis}[
    width=0.76\textwidth,
    height=0.44\textwidth,
    xlabel={Recursion level \(r\)},
    ylabel={Gap \(\alpha_r-2\)},
    xmin=0,
    xmax=100,
    ymin=0.008,
    ymax=0.45,
    xtick={0,10,20,30,40,50,60,70,80,90,100},
    grid=major,
    grid style={draw=gray!20},
    axis line style={draw=black!70},
    tick label style={font=\small},
    label style={font=\small},
]
\addplot[
    green!48!black,
    thick,
    smooth,
    domain=0:100,
    samples=300,
]
    {1/(x + 1/(2.371339-2))};

\addplot[
    green!48!black,
    only marks,
    mark=*,
    mark size=1.5pt,
    samples at={0,1,...,10,12,14,...,20,25,30,...,100},
]
    {1/(x + 1/(2.371339-2))};
\end{semilogyaxis}
\end{tikzpicture}
\caption{The convergence of
\(\alpha_r=2+1/(r+1/(\widetilde \omega-2))\) to \(2\), where
\(\widetilde \omega=2.371339 > \omega\) \cite{ADWXXZ24}.  The vertical axis shows the gap
\(\alpha_r-2\) on a logarithmic scale.}
\label{fig:alpha-convergence}
\end{figure}

\begin{theorem}[Deterministic sparsification]
\label{thm:deterministic_sparsification}
Let \(G\) be a positively weighted \(n\)-vertex, \(m\)-edge multigraph, where
\[
    b
    :=
    1+\left\lceil
        \log_2\frac{w_{\max}}{w_{\min}}
    \right\rceil.
\]
Let \(0<\eps\le1/2\), and suppose \(m,\eps^{-1}\le n^{O(1)}\).  There is a
deterministic algorithm that constructs a
\((1\pm\eps)\)-spectral sparsifier with
\begin{equation}
\label{eq:final-size-phi}
    O\!\left(
        n (b+\log^2(n))
        \phi_n^{-2}\eps^{-2}\log^4(n)
    \right)
\end{equation}
edges in
\begin{equation}
\label{eq:final-time-phi}
    m^{1+o(1)}
    +
    O\!\left(
        n^2 (b+\log^2(n))
        \phi_n^{-5}\eps^{-9/2}\log^{19/2}(n)
    \right)
\end{equation}
time.  Since
\(\phi_n^{-1}=\log^{9+o(1)}n\), these bounds are
\begin{equation}
\label{eq:final-size-logs}
    O\!\left(
        n (b+\log^2(n))
        \eps^{-2}\log^{22+o(1)}n
    \right)
\end{equation}
edges and
\begin{equation}
\label{eq:final-time-logs}
    m^{1+o(1)}
    +
    O\!\left(
        n^2 (b+\log^2(n))
        \eps^{-9/2}\log^{109/2+o(1)}n
    \right)
\end{equation}
time.
\end{theorem}

\begin{proof}
Discard isolated vertices. Set \(r:=\lceil\ln n\rceil\) and
\(\eta:=\eps/(4(r+1))\).  Then
\((1-\eta)^{r+1}\ge1-\eps\), while
\((1+\eta)^{r+1}\le e^{\eps/4}\le1+\eps\).  Moreover,
\cref{lem:alpha-closed-form} gives
\[
    n^{\alpha_r}
    =
    n^2n^{1/(r+1/(\omega-2))}
    =O(n^2).
\]

We have \(r+1=O(\log n)\),
\(\eta^{-2}=O(\eps^{-2}\log^2 n)\), and
\(\eta^{-9/2}=O(\eps^{-9/2}\log^{9/2}n)\).  Substitution into
\cref{thm:r-level} proves \eqref{eq:final-size-phi} and
\eqref{eq:final-time-phi}.  Since \(G\) has no isolated vertices,
\(n\le2m\), so the factor \(r+1\) multiplying the input-sensitive work is
absorbed by \(m^{1+o(1)}\).  Finally, substituting
\(\phi_n^{-1}=\log^{9+o(1)}n\) gives
\eqref{eq:final-size-logs} and \eqref{eq:final-time-logs}.
\end{proof}

\paragraph{Acknowledgments.}
The first author would like to thank Sivakanth Gopi, Janardhan Kulkarni, Yang P. Liu, Jakub Tarnawski, and Sam Wong for early discussions on deterministic graph sparsification. The second author used OpenAI’s GPT 5.6 Sol on Max effort to assist with drafting and revising portions of the exposition. The authors assume responsibility for all content.

\bibliographystyle{alpha}
\bibliography{references}

\appendix

\section{Inverse-Square-Root Approximation}
\label{app:inverse-square-root}

\inversesqrtapproximation*

\begin{proof}
We use the following standard consequence of the Chebyshev construction
for approximating reciprocals; see \cite{SV14}.  For
\(0<\ell\le u\) and \(0<\delta<1\), there is an explicitly computable
polynomial \(q\) of degree
\[
    O\!\left(
        \sqrt{\frac{u}{\ell}}\log\frac1\delta
    \right)
\]
such that \(|1-yq(y)|\le\delta\) for every \(y\in[\ell,u]\).

We also use the integral representation
\[
    x^{-1/2}
    =
    \frac1\pi\int_0^\infty
        \frac{t^{-1/2}}{x+t}\,dt.
\]
Indeed, substituting \(t=xs\) reduces the integral to
\(x^{-1/2}\pi^{-1}\int_0^\infty s^{-1/2}/(1+s)\,ds=x^{-1/2}\).

Set \(T:=16\tau^{-2}\) and \(\zeta:=\tau^2/64\).  For each
\(t\in[0,T]\), apply the reciprocal approximation with
\[
    \ell=a+t,\qquad
    u=2+t,\qquad
    \delta_t=(a+t)\zeta.
\]
Notice that \(\delta_t<1\).  This gives a polynomial \(q_t\) satisfying,
for every \(x\in[a,2]\),
\[
    \left|
        q_t(x+t)-\frac1{x+t}
    \right|
    \le
    \frac{\delta_t}{x+t}
    \le\zeta.
\]
Moreover,
\[
    \frac{2+t}{a+t}\le\frac2a,
    \qquad
    \log\frac1{\delta_t}
    \le
    \log\frac1{a\zeta},
\]
so all the \(q_t\)'s have degree at most
\[
    D
    =
    O\!\left(
        a^{-1/2}\log\frac{1}{a\tau}
    \right).
\]

Define
\[
    p(x)
    :=
    \frac1\pi\int_0^T t^{-1/2}q_t(x+t)\,dt.
\]
After padding the \(q_t\)'s with zero coefficients, the integrand is a
degree-\(D\) polynomial in \(x\); integrating its coefficients therefore
shows that \(p\) is also a degree-\(D\) polynomial.  For
\(x\in[a,2]\), the reciprocal-approximation error is at most
\[
    \frac1\pi\int_0^T t^{-1/2}\zeta\,dt
    =
    \frac{2\zeta\sqrt T}{\pi}
    =
    \frac{\tau}{8\pi}.
\]
The truncation error satisfies
\[
    \frac1\pi\int_T^\infty
        \frac{t^{-1/2}}{x+t}\,dt
    \le
    \frac1\pi\int_T^\infty t^{-3/2}\,dt
    =
    \frac{2}{\pi\sqrt T}
    =
    \frac{\tau}{2\pi}.
\]
Their sum is less than \(\tau\), proving the uniform approximation
guarantee.  The Chebyshev construction is explicit, and its coefficient
integrals can be evaluated to any required precision in time; decreasing the
internal accuracy by a constant factor absorbs this numerical error
without changing the degree bound.
\end{proof}

\end{document}